\documentclass[11pt,a4paper]{article}

\newcommand{\Sect}{{Sect.~}}

\usepackage{amsmath, amsfonts, amssymb, amsthm}
\usepackage{mathrsfs}
\usepackage{physics}
\usepackage{braket}
\usepackage{enumitem}
\usepackage[mode=buildnew]{standalone}
\usepackage{tikz} 
\usetikzlibrary{matrix,arrows,patterns,cd}
\usepackage{caption}
\usepackage{subcaption}
\usepackage{float}
\usepackage{graphicx}
\usepackage{accents}
\usepackage[margin=2.2cm]{geometry}

\newtheorem{mydef}{Definition}[section]

\newtheorem{theo}[mydef]{Theorem}
\newtheorem{lem}[mydef]{Lemma}

\newtheorem{corol}[mydef]{Corollary}

\usepackage{scalerel}
\usetikzlibrary{svg.path}

\definecolor{orcidlogocol}{HTML}{A6CE39}
\tikzset{
  orcidlogo/.pic={
    \fill[orcidlogocol] svg{M256,128c0,70.7-57.3,128-128,128C57.3,256,0,198.7,0,128C0,57.3,57.3,0,128,0C198.7,0,256,57.3,256,128z};
    \fill[white] svg{M86.3,186.2H70.9V79.1h15.4v48.4V186.2z}
                 svg{M108.9,79.1h41.6c39.6,0,57,28.3,57,53.6c0,27.5-21.5,53.6-56.8,53.6h-41.8V79.1z M124.3,172.4h24.5c34.9,0,42.9-26.5,42.9-39.7c0-21.5-13.7-39.7-43.7-39.7h-23.7V172.4z}
                 svg{M88.7,56.8c0,5.5-4.5,10.1-10.1,10.1c-5.6,0-10.1-4.6-10.1-10.1c0-5.6,4.5-10.1,10.1-10.1C84.2,46.7,88.7,51.3,88.7,56.8z};
  }
}

\newcommand\orcidicon[1]{\href{https://orcid.org/#1}{\mbox{\scalerel*{
\begin{tikzpicture}[yscale=-1,transform shape]
\pic{orcidlogo};
\end{tikzpicture}
}{|}}}}
\usepackage{hyperref}

\DeclareMathOperator{\Prob}{Prob}

\DeclareMathOperator{\openone}{1\!\!\!\!1}
\DeclareMathOperator{\Ran}{Ran}
\DeclareMathOperator{\supp}{supp}
\DeclareMathOperator{\Weyl}{Weyl}

\usepackage[numbers]{natbib}

\title{Asymptotic measurement schemes for every observable of a quantum field theory\hspace{2mm}}
\author{Christopher J. Fewster$^{1,a}$ \orcidicon{0000-0001-8915-5321}, Ian Jubb$^{2,b}$ \orcidicon{0000-0001-7339-2058}, Maximilian~H.~Ruep$^{1,c}$ \orcidicon{0000-0001-6866-4506}\vspace{4mm}\\
{\small $^1$Department of Mathematics,}\\
{\small University of York, Heslington, York YO10 5DD, United Kingdom}\vspace{2mm}\\
{\small $^2$School of Theoretical Physics,}\\
{\small Dublin Institute for Advanced Studies, 10 Burlington Road, Dublin 4, Ireland}\vspace{4mm}\\
{\footnotesize ~$^a$~\href{mailto:chris.fewster@york.ac.uk}{chris.fewster@york.ac.uk}~,~$^b$~\href{mailto:ijubb@stp.dias.ie}{ijubb@stp.dias.ie}~,~$^c$~\href{mailto:maximilian.ruep@york.ac.uk}{maximilian.ruep@york.ac.uk}}}
\date{\today}

\begin{document}

\maketitle

\begin{abstract}
In quantum measurement theory, a measurement scheme describes how an observable of a given system can be measured indirectly using a probe. 
The measurement scheme involves the specification of a probe theory, an initial probe state, a probe observable and a coupling between the system and the probe, so that a measurement of the probe observable after the coupling has ceased reproduces (in expectation) the result of measuring the system observable in the system state. Recent work has shown how local and causal measurement schemes may be described in the context of model-independent quantum field theory (QFT), but has not addressed the question of whether such measurement schemes exist for all system observables. Here, we present two treatments of this question. The first is a proof of principle which provides a measurement scheme for every local observable
of the quantized real linear scalar field if one relaxes one of the conditions on a QFT measurement scheme by allowing a non-compact coupling region. Secondly, restricting to compact coupling regions, we explicitly construct \emph{asymptotic} measurement schemes for every local observable of the quantized theory. More precisely, we show that for every local system observable $A$ there is an associated collection of measurement schemes for system observables that converge to $A$. All the measurement schemes in this collection have the same fixed compact coupling zone and the same processing region. The convergence of the system observables holds, in particular, in GNS representations of suitable states on the field algebra or the Weyl algebra. In this way, we show that every observable can be asymptotically measured using locally coupled probe theories.
\end{abstract}

\vspace{2pc}
\noindent{\it Keywords}: induced observables, impossible measurements, local probes, local measurement schemes, asymptotic measurement schemes

\section{Introduction}

The theoretical model of a system in quantum theory has two main ingredients: observables and states. Together, they enable the calculation of predictions for the results of potential measurements of specified observables in specified states -- at least in the statistical sense of predicting expectation values and higher moments for the results from an ensemble of identical measurement runs. What is purposefully left out by this minimal description is an account of the actual measurement process in which (individual) results are read out to sufficient accuracy using apparatus coupled to the system. This is the task of quantum measurement theory (QMT)~\cite{Busch2016}. 

While currently there seems to be no full explanation of the measurement process in reach, QMT has achieved an understanding of individual steps along the \emph{measurement chain}, i.e., the process by which information about the quantum system may be transferred to a \emph{quantum} probe and in principle be extracted by a measurement of probe observables. Concretely, for a given system $\mathcal{S}$ one considers an auxiliary quantum structure called the \emph{probe} $\mathcal{P}$ together with an initial probe state $\sigma$, a ``pointer'' or probe observable $B$ and $\mathcal{C}$, a ``measurement coupling'' to the system, see Chapter 10 in~\cite{Busch2016}. The idea is that a measurement of $B$ after the coupling has been removed will yield information about a system observable\footnote{In the following we will distinguish between a general \emph{element} of a system or probe algebra and an \emph{observable}, which is a \emph{Hermitian} element.} 
which we say has been ``induced'' by $B$ and denote by $\varepsilon_{\sigma}^\mathcal{C}(B)$. More technically,
one requires that the expectation value of $\varepsilon_{\sigma}^\mathcal{C}(B)$ in any initial system state $\omega$ should equal the expectation value of the observable $\openone \otimes B$ of the combined system-probe structure $\mathcal{S}\otimes \mathcal{P}$ after the coupling, if its state beforehand was $\omega\otimes\sigma$. The collection of probe $\mathcal{P}$, probe observable $B$, initial probe state $\sigma$ and coupling $\mathcal{C}$ forms a (Hermitian)\footnote{In the formal definition of a measurement scheme we will allow general, i.e., not necessarily Hermitian elements $B$ of the probe algebra, which may induce non-Hermitian elements of the system algebra. A measurement scheme for a Hermitian element using a Hermitian $B$ will be called \emph{Hermitian}.} \emph{measurement scheme} for the system observable $\varepsilon_{\sigma}^\mathcal{C}(B)$.

It is important to emphasize that this description neither contains nor requires an explanation of \emph{how} to extract information from the probe. What needs to be put in is the standard working assumption that information (i.e., the expectation value of $\openone \otimes B$) can be observed \emph{somehow}. 

A central issue is to determine what observables of a quantum system can be measured in such a way. One can of course scan through all possible probe observables and states to classify the system observables that are induced by a given coupled probe. However, it is certainly \emph{more} interesting to ask whether \emph{every} system observable can be measured by \emph{some} probe, which is a somewhat inverse problem. In the quantum mechanical setting it turns out that by allowing arbitrary unitary interactions between system and probe, the Stinespring dilation theorem implies ``the first fundamental theorem of the quantum theory of measurement''~\cite{Busch2016}, i.e., that for every system observable there exists a measurement scheme, see Theorem 10.1 in~\cite{Busch2016}. We emphasize that the Stinespring theorem does not address the question of whether the required unitary interaction is physically reasonable.

While well-understood in non-relativistic quantum mechanics, measurement schemes have only recently been discussed in the context of relativistic quantum fields on a possibly curved globally hyperbolic spacetime $M$ in~\cite{fewster2018quantum}, which we call the FV framework (after the authors of~\cite{fewster2018quantum}). Crucially, the fundamental principle of the FV framework that interactions between system and probe quantum \emph{fields} should be \emph{local} puts restrictions on the possible couplings and yields fully local and covariant measurement schemes for local system observables.
The physically reasonable constraints on the possible couplings in the FV framework and the consequent lack of suitable dilation-type results raises the important question whether it still holds true that every system observable of a quantum field can be measured by a \emph{local} measurement scheme in the FV framework.

This paper will address and answer that question for a system consisting of a quantized linear real scalar field, described either in terms of a $*$-algebra generated by smeared fields, or using the Weyl $C^*$-algebra quantization. We first discuss an explicit coupling to a probe theory that provides a measurement scheme for every local algebra element. However, this example is not fully satisfactory, because the coupling region is non-compact (in spacetimes with non-compact Cauchy surfaces) and therefore does not fully comply with the FV framework; nonetheless, it may be viewed as a proof of principle. In order to address the general situation with compact coupling regions, we then introduce the notion of a (Hermitian) \emph{asymptotic} measurement scheme for an observable $A$, i.e., a collection of measurement schemes for the system observables $\varepsilon_{\sigma_\alpha}^{\mathcal{C}_\alpha}(B_\alpha)$ such that $\varepsilon_{\sigma_\alpha}^{\mathcal{C}_\alpha}(B_\alpha)$ converges to $A$ with respect to a topology on $\mathcal{S}$ and thus provides a way of measuring $A$ to arbitrary precision. We will prove that there is an asymptotic measurement scheme for every local system element. In fact, we show that there is a \emph{Hermitian} asymptotic measurement scheme for every local system \emph{observable}, which equivalently shows that a dense set of local system observables have Hermitian measurement schemes. More concretely for every precompact\footnote{Here, a precompact subset is one that has a compact closure.} region $N \subseteq M$, every system observable localizable in $N$ and every $L \subseteq M\setminus J^-(\overline{N})$ whose domain of dependence contains $N$, there exists a Hermitian asymptotic measurement scheme for $A$ with coupling in $N$ and processing region $L$ (see Fig.~\ref{fig:measurement_scheme_setup} for an illustration of this spacetime setup). Here a processing region is a region in which an experimenter needs to have control over their probe in order to read out the desired information\footnote{This concept of a processing region was introduced in~\cite{Ruep2021} in the context of entanglement harvesting.}. The convergence holds in particular in the strong$^*$ operator topology of the GNS representation of any quasi-free state with distributional two-point function.

\begin{figure}[ht]
    \centering
    \includegraphics[width=0.75\textwidth]{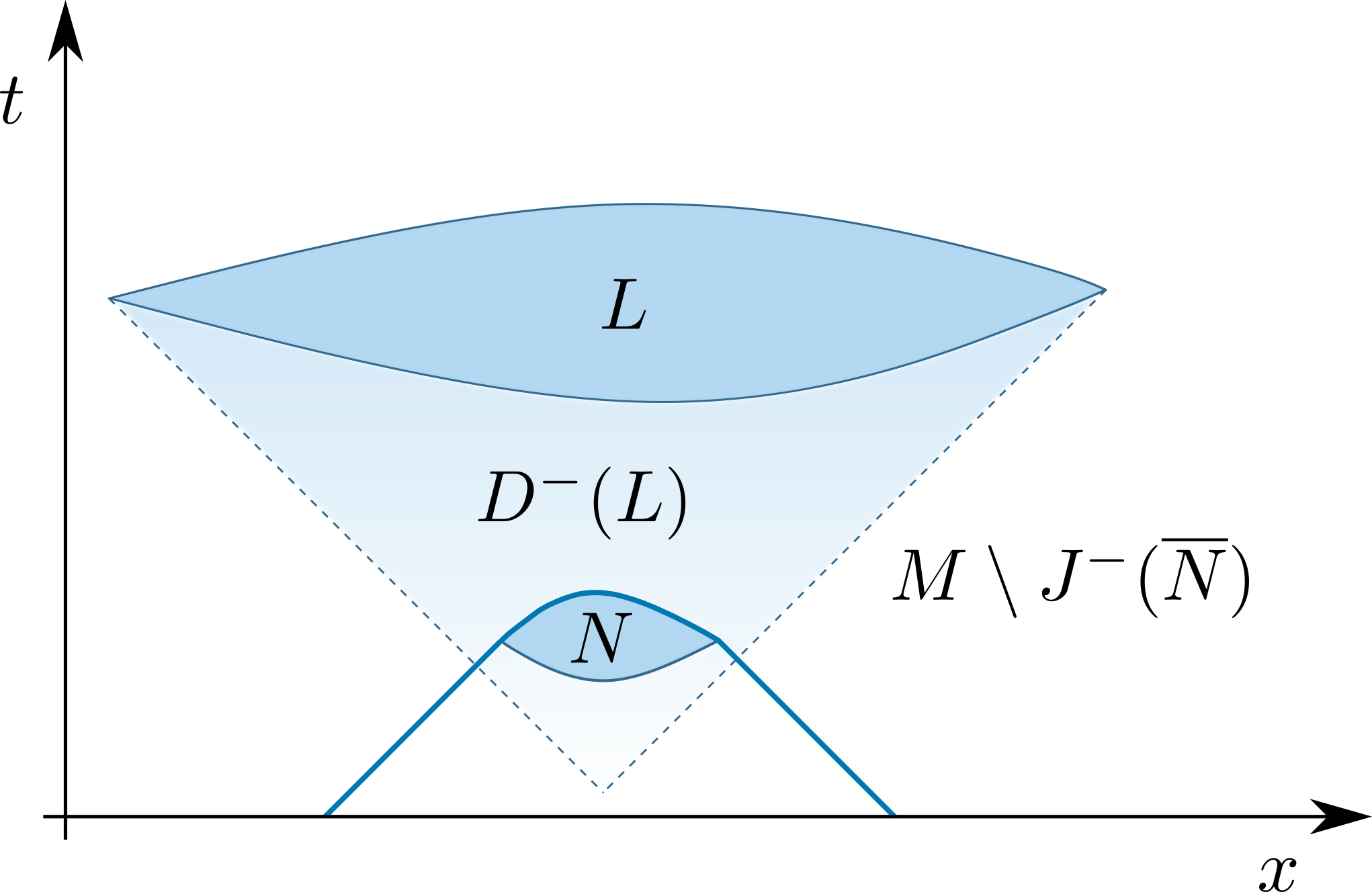}
    \caption{$1+1$-dimensional illustration of our spacetime setup. To measure a system observable in $N$ we couple the probe to the system in that region, and measure some probe observable in the processing region $L$. The dotted line illustrates the past domain of dependence of $L$, denoted as $D^-(L)$, which must contain $N$. $L$ must also not be to the past of the closure of $N$ (that is, $L\subseteq M\setminus J^-(\overline{N})$). In the figure, $M\setminus J^-(\overline{N})$ consists of those points that are above the solid line.}
    \label{fig:measurement_scheme_setup}
\end{figure}\

Let us now turn to the intuition behind our result, which is most easily exemplified by considering the field algebra. For every local observable of a linear real scalar system field $\varphi_\mathcal{S}$ we will construct a sequence of inducible observables that converges (in a specified rigorous sense) to the desired system observable. For definiteness, let us consider a system observable given as a smeared field operator $\varphi_\mathcal{S}(f)$ for a real-valued smooth function $f$ compactly supported in $N$. Let $\varphi_\mathcal{P}$ be a linear real scalar probe field, with initial state $\sigma$, and consider the coupled theory given by the following quadratic Lagrangian density
\begin{equation}
    \begin{aligned}
    \mathcal{L}_\mathcal{S} + \mathcal{L}_\mathcal{P} - {\lambda}\rho \varphi_\mathcal{S} \varphi_{\mathcal P},
    \end{aligned}
    \label{eq_lagrangians_intro}
\end{equation}
where $\rho$ is a real-valued smooth coupling function compactly supported in $N$ and $\lambda$ is a coupling constant; $\mathcal{L}_\mathcal{S}$ and $\mathcal{L}_\mathcal{P}$ being the Lagrangian densities of the system and probe fields individually. Let $h$ be a smooth function compactly supported in $L$, and so outside the causal past of the support of $\rho$. We will show that for a $\lambda$-dependent probe observable $ B_\lambda=\varphi_\mathcal{P}(h/\lambda) +  c_{\sigma,\lambda}\openone$, where $c_{\sigma,\lambda}$ are explicitly known real numbers, the induced system observable is given by
\begin{equation}
    \varepsilon_{\sigma}^{ \mathcal{C}_\lambda }\qty(B_\lambda) = \varphi_\mathcal{S}(-\rho E_P^- h) + \mathcal{O}(\lambda^2) \; \overset{\lambda \to 0}{\longrightarrow} \;  \varphi_\mathcal{S}(-\rho E_P^- h),
\end{equation}
where $E_P^-$ is the advanced Green operator of the probe theory. Given any test function $f$, we will 
reverse-engineer $\rho$ and $h$ so that $\varphi_\mathcal{S}(f)=\varphi_\mathcal{S}(-\rho E_P^- h)$,  whereupon the limit $\varepsilon_{\sigma}^{ \mathcal{C}_\lambda}\qty(B_\lambda)\to \varphi_S(f)$ exists in a specific rigorous sense, thus showing that any smeared field may be approximated by a sequence of inducible observables. 

To describe this mathematical procedure in physical terms: one first tunes the apparatus and pointer observables (i.e., a probe theory and a probe observable respectively) to the desired system observable of interest, i.e., one chooses the coupling function $\lambda \rho$ and the probe test function $h/\lambda$ dependent on the observable $\varphi_\mathcal{S}(f)$ such that $\varphi_\mathcal{S}(-\rho E_P^- h) = \varphi_\mathcal{S}(f)$, then one lets $\lambda$ decrease, i.e., one reduces the coupling strength while simultaneously increasing the ``sensitivity'' of the apparatus.

As it turns out, there is a precise mathematical way of quantifying the efficiency of an asymptotic measurement scheme relative to the effort of increasing the apparatus' sensitivity. The effort of the described asymptotic measurement scheme for $\varphi_\mathcal{S}(f)$ can be quantified by the square root of the variance of the probe observables $B_\lambda$ in the state $\sigma$, which \emph{diverges} as $\lambda \to 0$. Concretely, it diverges as $\lambda^{-1}$, while $\varphi_\mathcal{S}(f) - \varepsilon_{\sigma}^{ \mathcal{C}_\lambda}\qty(B_\lambda) = \mathcal{O}(\lambda^2)$. This motivates us to introduce the \emph{order} of an asymptotic measurement scheme given by the rate of convergence of $\varphi_\mathcal{S}(f) - \varepsilon_{\sigma}^{ \mathcal{C}_\lambda}\qty(B_\lambda)$ relative to the rate of divergence of $B_\lambda$. The present Hermitian asymptotic measurement scheme using a single probe fields is hence of order $2$, and we will show that every $\varphi_\mathcal{S}(f)$ admits an asymptotic measurement scheme of any arbitrary even order $2k$. We achieve this by utilizing $k$ probe fields.

Moving on, we show that the coupling in Eq.~\eqref{eq_lagrangians_intro} and an appropriate choice of probe elements allows one to construct asymptotic measurement schemes for any power $\varphi_\mathcal{S}(f)^n$ as well. Furthermore, by using $k$ probe fields, these measurement schemes can be combined to form an asymptotic measurement scheme for an arbitrary $k$-fold complex linear combination of $\varphi_\mathcal{S}(f_j)^{n_j}$. Allowing arbitrary $k$ and using multilinear polarization, we hence get asymptotic measurement schemes for every element in the field algebra. In fact, we construct Hermitian asymptotic measurement schemes for every observable. Similar arguments apply to ``exponentiated fields'', i.e., Weyl generators, complex linear combinations thereof and ultimately also their $C^*$-closure, i.e., the full Weyl algebra. 

We therefore show the existence of (Hermitian) asymptotic measurement schemes for every (Hermitian) element of the field algebra and the Weyl algebra, supporting the interpretation of Hermitian elements as local \emph{observables}.\footnote{Note that non-Hermitian elements of either the field or the Weyl algebra are certainly not induced by \emph{Hermitian} probe elements. However, a Hermitian element could be induced by a non-Hermitian element.} Let us emphasize that it was recently shown in~\cite{bostelmann2020impossible} that measurement schemes in the FV framework, and hence also the measurement schemes in the present manuscript, are \emph{causal}. In particular, the state update rules following non-selective\footnote{Selective measurements require joint post-selection of the experimental results of all experimenters and are hence causal \emph{by fiat}.} measurements do not result in superluminal signalling. Moreover, they yield consistent results if some or all of the individual measurements are combined and
also do not depend on a specific causal ordering of the measurements if more than one is possible.
The compatibility of measurement with causality is a non-trivial result, as arguments put forward long ago by Sorkin~\cite{sorkin1993impossible} suggest that many conceivable local quantum channels on quantum fields are in conflict with causality, in particular state-update rules associated to ideal measurements~\cite{Beckman2002}. Sorkin's argument has since been refined in~\cite{Borsten2021,jubb2021causal}, where general conditions were given for a local quantum channel to be causally consistent. In~\cite{jubb2021causal} several standard examples of local quantum channels were considered in real scalar QFT, including action on the state by a local unitary, and it was explicitly shown that `most' of these standard quantum channels do indeed violate causality. Consequently, they cannot be the update rules for local measurement schemes of the FV type.\footnote{In general, exactly repeatable updates for continuous observables cannot be realized by measurement schemes in any case, see~\cite{OkamuraOzawa2016} (in particular the penultimate paragraph of \Sect~V therein).}

These results provide a sharp perspective on the interpretation of elements of local algebras in algebraic quantum field theory. According to the original interpretation of Haag and Kastler, elements of a local algebra $\mathcal{A}(R)$ are interpreted ``as representing physical operations performed in the region'' $R$~\cite{HaagKastler1964}. Here, the operation associated with element $A\in\mathcal{A}(R)$ and $\omega(A^*A)>0$ 
is the quantum channel
\begin{equation}
    \omega\mapsto \omega_A, \qquad \omega_A(B)= \frac{\omega(A^*BA)}{\omega(A^*A)}.
\end{equation}
However, for many local algebra elements, $A$, the above channel is acausal and cannot be obtained by physically realizable processes. See for instance the explicit example in \Sect III.~C. in~\cite{jubb2021causal}.
Such quantum channels, then, cannot be interpreted as the mathematical description of any physical process. Going further, for many local Hermitian elements, $A$, the `standard' quantum channels associated to them, (such as the appropriate action of the unitary $e^{iA}$, or projective or Gaussian measurements) are acausal as well, see \Sect IV. E. in~\cite{jubb2021causal}.

By contrast the results of this paper show that every Hermitian local algebra element is observable via a Hermitian \emph{asymptotic} FV measurement scheme in a causally consistent manner\footnote{This viewpoint is also reinforced by~\cite{jubb2021causal}, where Gaussian/weak measurements of smeared fields are used to extract expectation values of many Hermitian elements of $\mathcal{A}(R)$ in a causal way.}. Therefore, it seems to us more appropriate to regard $\mathcal{A}(R)$ as an algebra of (or rather generated by) local observables than an algebra of local operations. This is by no means to deny that the operations associated with \emph{physical}, i.e., non-signalling, processes can have an important role in the theory -- for instance, 
the recent $C^*$-approach to interacting QFTs of Buchholz and Fredenhagen~\cite{BuchholzFredenhagen2020} is based exactly on algebras generated by causal local unitary operators associated with local interaction Lagrangians. Rather, we are concerned with what the interpretation of a \emph{general} element of a local algebra should be.

The detailed structure of our paper is as follows.
In \Sect\ref{sec_asympt_ms} we give more detail about the FV framework and introduce asymptotic measurement schemes. As a proof of principle, measurement schemes (with possibly non-compact coupling zone) for every observable of a massless free scalar field are constructed in \Sect\ref{sec:proof_of_p}. The main ideas for the construction of asymptotic measurement schemes for the field and the Weyl algebra are introduced in \Sect\ref{sec:classical} for \emph{classical fields}, and then applied to the respective quantized versions in \Sect\ref{sec:asymp_m_schemes_field_algebra} and \Sect\ref{sec_qams_Weyl}. In \Sect\ref{sec_variance} we discuss the physical interpretation of the approximation procedure, linking our notion of effort to the resources required to make a measurement to specified accuracy before we conclude in \Sect\ref{sec:discussion}.

\section{Asymptotic measurement schemes in the FV framework}\label{sec_asympt_ms}

As already mentioned, the FV framework implements measurement schemes in relativistic QFT, where both the system $\mathcal{S}$ as well as the probe $\mathcal{P}$ are QFTs, possibly generated by several quantum fields. The coupling between the system and probe is constrained in two ways. First, it is confined to a compact \emph{coupling zone} $K$ of spacetime, acknowledging that
measurements are bounded both in time and space. Second, the coupled theory  $\mathcal{C}$ of system and probe should itself be a \emph{bona fide} QFT. Under these circumstances, we say that the \emph{coupled combination} $\mathcal{C}$ is a coupled variant of the \emph{uncoupled combination} $\mathcal{S} \otimes \mathcal{P}$. (In the case of classical theories -- which we will also discuss below -- the uncoupled combination is described by a direct sum.)

Since the interaction is only active in the compact coupling zone $K$, there are natural `in' and `out' regions -- given by the complements of the causal future and past of $K$ respectively -- in which
the coupled theory $\mathcal{C}$ may be identified with the uncoupled combination $\mathcal{S} \otimes \mathcal{P}$. These identifications can be used to define a scattering automorphism $\Theta$ of $\mathcal{S}\otimes \mathcal{P}$ with the following action: using the future identification (associated with the `out' region), an observable $X$ of the uncoupled combination $\mathcal{S}\otimes \mathcal{P}$ is mapped to an observable of the coupled theory $\mathcal{C}$, which is then mapped back to $\mathcal{S}\otimes\mathcal{P}$ using the past identification (associated with the `in' region) to give $\Theta X$. Put differently, the action of $\Theta$ is  morally speaking the adjoint action of a putative unitary scattering matrix arising from M\o ller operators. 

Let us now sketch the construction of the induced observables from~\cite{fewster2018quantum}. Consider an experiment in which the system and probe are prepared independently in states $\omega$ and $\sigma$ respectively, in the `in' region before the coupling takes effect, and a probe observable $B$ is measured in the `out region', i.e., after the coupling has ceased. To calculate the expected outcome, we first express
$B\in \mathcal{P}$ as the observable $\openone\otimes B\in \mathcal{S}\otimes\mathcal{P}$, which in turn may be identified with an observable  $\widetilde{B}\in\mathcal{C}$ using the future identification. The expected experimental outcome is the expectation of $\widetilde{B}$ in the state $\underaccent{\tilde}{\omega}_\sigma$ obtained from $\omega\otimes\sigma$ using the past identification, and is given by the formula 
\begin{equation}
\underaccent{\tilde}{\omega}_\sigma(\widetilde{B})=(\omega\otimes\sigma)(\Theta\openone\otimes B).
\end{equation}
We may interpret this experiment as a measurement scheme for a system observable
$\varepsilon_{\sigma}^{\mathcal{C}}(B)$ with the property
\begin{equation}
\omega(\varepsilon_{\sigma}^{\mathcal{C}}(B)) = \underaccent{\tilde}{\omega}_\sigma(\widetilde{B})
\end{equation}
and which is given explicitly by 
\begin{equation}
\varepsilon_{\sigma}^{\mathcal{C}}(B) = \eta_\sigma(\Theta \; {\openone} \otimes B),
\label{eq_induced_obs}
\end{equation}
where $\eta_\sigma$, is determined by the formula $\eta_\sigma(A\otimes B) = \sigma(B)A$, linearity  and (where relevant) continuity,
so that $\omega(\varepsilon_{\sigma}^{\mathcal{C}}(B))=
(\omega \otimes \sigma) (\Theta \openone \otimes B) =\underaccent{\tilde}{\omega}_\sigma(\widetilde{B})$ for every system state $\omega$. 
The map $\varepsilon_{\sigma}^{\mathcal{C}}$ describing the observables induced by this coupling may be extended to all probe elements as a completely positive, unit-preserving linear map from $\mathcal{P}$ to $\mathcal{S}$. Although this scheme is very general, it is amenable to concrete calculations in simple cases. 

The central issue explored in this paper is to determine the set of system observables that can be measured by FV measurement schemes, allowing for different probe preparation states and indeed different probe theories and couplings. The ideal situation would be if all of the system observables\footnote{It might be the case that one wishes to regard as observables only those Hermitian elements that are also invariant under a global gauge group. It has been shown in~\cite{fewster2018quantum} that (as one would hope) if the coupling and probe preparation state are gauge-invariant then all elements induced by observables, i.e., Hermitian and gauge invariant elements of the probe, are themselves observables. One would not necessarily expect a theory of measurement to extend to unobservable quantities!} admit a Hermitian measurement scheme, i.e., are \emph{inducible} by observables. As mentioned above, we cannot invoke the Stinespring dilation theorem, because it is crucial that the probe and coupled theories are also local QFTs. Therefore we will content ourselves with the situation in which the system, probe, uncoupled and coupled theories are all linear quantum field theories described by quadratic Lagrangians. In particular, the coupling between system and probe is mediated by quadratic interaction terms in the Lagrangian. Nevertheless such simple interactions turn out to be sufficient for our purposes.

In the following we will show that for any local element $A$ of the system theory, there is a family of measurement schemes whose induced elements converge to $A$.
In particular, the set of inducible elements is dense in the system theory. To make this precise we introduce the concept of an \emph{asymptotic measurement scheme}.\footnote{The term ``asymptotic measurement scheme'' also appears in~\cite{KiukasLahti2008}. Their notion however is different from the asymptotic measurement schemes here.} To minimise notation, we will represent a measurement scheme for an element $A\in \mathcal{S}$ by a triple $H=(\mathcal{P},\varepsilon_\sigma,B)$ where $\mathcal{P}$ is a probe theory, $\varepsilon_{\sigma}$ is the induced observable map\footnote{Note the slight abuse of terminology here: Not every element in the range of the ``induced observable map'' is Hermitian!} emerging from a coupled combination of $\mathcal{S} \otimes \mathcal{P}$ and some probe state $\sigma$, and $B\in \mathcal{P}$ is a probe element, so that $A=\varepsilon_\sigma(B)$; we thus suppress the coupling zone $K$, the coupled combination $\mathcal{C}$ of $\mathcal{S} \otimes \mathcal{P}$, the identification maps and the resulting scattering map $\Theta$, which defines $\varepsilon_\sigma$. The formal definition of an asymptotic measurement scheme is now as follows.

\begin{mydef}[$\tau$-asymptotic measurement scheme]
Let $\mathcal{S}$ be a system theory equipped with a topology $\tau$. 
An \emph{asymptotic measurement scheme} for $A \in  \mathcal{S}$ with respect to $\tau$ 
(also called a $\tau$-asymptotic measurement scheme) is a collection $(H_\alpha)_{\alpha \in J}$, where $J$ is a directed set and each 
$H_\alpha=\qty(\mathcal{P}_\alpha, \varepsilon_{\alpha, \sigma_\alpha}, B_\alpha)$
is a measurement scheme for $\varepsilon_{\alpha, \sigma_\alpha}(B_\alpha) \in \mathcal{S}$, such that
\begin{equation}
    \varepsilon_{\alpha, \sigma_\alpha}(B_\alpha) \to A
\end{equation}
with respect to $\tau$. We say that an asymptotic measurement scheme $(H_\alpha)_{\alpha \in J}$ has \emph{coupling in $N$} and \emph{processing region $L$}, if for every $\alpha \in J$ the coupling zone of $H_\alpha$ is contained in $N$ and $B_\alpha \in \mathcal{P}_\alpha(L)$. We say that $(H_\alpha)_{\alpha \in J}$ is \emph{Hermitian} if for every $\alpha:$ $B_\alpha^* = B_\alpha$.
\label{def_ams}
\end{mydef}

Notice that the convergence above is in general understood in terms of nets, but of course includes the simpler situation of sequential convergence. 

A few remarks are in order: \begin{enumerate} 
    \item An element admiting an [asymptotic] measurement scheme may be called [asymptotically] inducible.
    \item The subscript of $\varepsilon_{\alpha, \sigma_\alpha}$ emphasises that it is a map from $\mathcal{P}_\alpha$ to $\mathcal{S}$ that depends on the state $\sigma_\alpha$ on $\mathcal{P}_\alpha$.
    \item It is clear that the definition of asymptotic measurement schemes in terms of limits of (appropriate) measurement schemes is not restricted to algebraic quantum field theory but equally applies to algebraic quantum theory more generally.
    \item If the topology $\tau$ is a vector space topology, i.e., if addition and scalar multiplication are continuous, then, by linearity of $\varepsilon_{\alpha, \sigma_\alpha}$, we see that whenever $\qty(\mathcal{P}_\alpha, \varepsilon_{\alpha, \sigma_\alpha}, B_\alpha)$ is an asymptotic measurement scheme for $A$, and $\qty(\mathcal{P}_\alpha, \varepsilon_{\alpha, \sigma_\alpha}, B'_\alpha)$ is an asymptotic measurement scheme for $A'$, 
    then for every $c \in \mathbb{C}$, $\qty(\mathcal{P}_\alpha, \varepsilon_{\alpha, \sigma_\alpha}, B_\alpha + c B'_\alpha)$ is an asymptotic measurement scheme for $A+c A'$. Notice that we require that both asymptotic schemes here share the same probes, couplings and preparation states, so this observation does not by itself imply that the set of \emph{all} inducible elements is a subspace.
    \item\label{real_part_ams} Similarly, if the topology $\tau$ is $*$-compatible, i.e., the $*$-operation is continuous with respect to $\tau$, then  whenever $\qty(\mathcal{P}_\alpha, \varepsilon_{\alpha, \sigma_\alpha}, B_\alpha)$ is an asymptotic measurement scheme for $A$,
    $\qty(\mathcal{P}_\alpha, \varepsilon_{\alpha, \sigma_\alpha}, B_\alpha^*)$ is an asymptotic measurement scheme for $A^*$, because $\varepsilon_{\alpha, \sigma_\alpha}(B_\alpha^*)=\varepsilon_{\alpha, \sigma_\alpha}(B_\alpha)^*$ by Theorem~3.2 in~\cite{fewster2018quantum}. In particular, if $A=A^*$ admits an asymptotic measurement scheme $\qty(\mathcal{P}_\alpha, \varepsilon_{\alpha, \sigma_\alpha}, B_\alpha)$, then
    $\qty(\mathcal{P}_\alpha, \varepsilon_{\alpha, \sigma_\alpha}, \tfrac{1}{2}(B_\alpha+B_\alpha^*))$ 
    is a Hermitian asymptotic measurement scheme for $A$: \emph{every asymptotically measurable observable is measurable by a net of observables}.
    \item Finally, if $\tau$ is a $*$-algebra topology, i.e., linear combinations, the $*$-operation and product are (jointly) continuous, and in which the cone of positive elements is $\tau$-closed,
    then, if $\qty(\mathcal{P}_\alpha, \varepsilon_{\alpha, \sigma_\alpha}, B_\alpha)$ is an asymptotic measurement scheme for $A$,  
     $\qty(\mathcal{P}_\alpha, \varepsilon_{\alpha, \sigma_\alpha}, B_\alpha^*B_\alpha)$ is an asymptotic measurement scheme for some $C$ with $A^*A\le C$, 
     again using Theorem~3.2 in~\cite{fewster2018quantum}. Here, we remind the reader that induced observable maps are generally not homomorphisms.   
\end{enumerate}

An elementary observation is:
\begin{lem}
Let $\mathcal{S}_m$ (respectively, $\mathcal{S}_a$) be the set of $A\in\mathcal{S}$ such that there is a
measurement scheme (resp., a $\tau$-asymptotic measurement scheme) for $A$. Then $\mathcal{S}_a$ is the closure of $\mathcal{S}_m$ in $\mathcal{S}$. Consequently $\mathcal{S}_a=\mathcal{S}$ if and only if $\mathcal{S}_m$ is dense in $\mathcal{S}$.
\label{lem_closed}
\end{lem}
\begin{proof}
Suppose $A \in \mathcal{S}_a$. Then, by definition, there exists a net of elements $A_\alpha$ in $\mathcal{S}_m$ that converges to $A$, i.e., $\mathcal{S}_a \subseteq \overline{\mathcal{S}_m}$. Conversely, if $A \in \overline{\mathcal{S}_m}$ then $A=\lim_\alpha A_\alpha$ where $(A_\alpha)_\alpha$ is a net of elements in $\mathcal{S}_m$. Accordingly, we may find a measurement scheme $H_\alpha$ for each $A_\alpha$, whereupon $(H_\alpha)_\alpha$ is a $\tau$-asymptotic measurement scheme for $A$. Hence $\overline{\mathcal{S}_m} \subseteq \mathcal{S}_a$.
\end{proof}

Before we set out to construct Hermitian asymptotic measurement schemes for \emph{every} observable of a linear real scalar field, we first give an explicit example in which every observable admits a \emph{bona fide} Hermitian measurement scheme, albeit by relaxing the condition that the coupling zone be compact.

\section{A proof of principle}
\label{sec:proof_of_p}

To start, we give a simple construction that -- although in general not fully respecting the demands of the FV framework -- nonetheless gives strong reason to believe that inducible observables should form a large subset of the system observables.

Consider a situation in which the system and probe are linear scalar fields of equal mass $m$, obeying the Klein--Gordon equation on a fixed globally hyperbolic spacetime. For the purposes of argument, we assume that there is some way of physically distinguishing the two species of scalar field. The uncoupled combination may be described conveniently as the theory of a complex scalar field
\begin{equation}
    \Phi = \frac{\varphi_\mathcal{S} + i\varphi_\mathcal{P}}{\sqrt{2}}
\end{equation}
also obeying the Klein--Gordon equation
\begin{equation}
   P\Phi:= (\Box +m^2)\Phi = 0
\end{equation}
Let $\chi\in C^\infty(M;\mathbb{R})$ and set $\Psi = e^{-i\chi}\Phi$. It is a standard exercise to show that $\Psi$ satisfies 
\begin{equation}\label{eq:ppQ}
    Q\Psi:= (D^\mu D_\mu + m^2)\Psi = 0
\end{equation}
where $D_\mu = \nabla_\mu + i A_\mu$ and $A_\mu=\nabla_\mu\chi$ is regarded as an external gauge potential. 
The advanced ($-$) and retarded ($+$) Green operators of $P$ and $Q$ are related by
\begin{equation}
    E_Q^\mp F= e^{-i\chi}E_P^\mp e^{i\chi}F
\end{equation}
where $F\in C_c^\infty(M;\mathbb{C})$. 

Suppose, more specifically, that $\chi$ vanishes identically to the future of Cauchy surface $\Sigma^+$, and takes the constant value $\pi/2$ to the past of Cauchy surface $\Sigma^-$. Then the gauge potential $A_\mu$ vanishes to the past of $\Sigma^-$ and the future of $\Sigma^+$, so $P$ and $Q$ agree except in the region between the Cauchy surfaces, $J^+(\Sigma^-)\cap J^-(\Sigma^+)$. With this in mind, we adopt $Q\Psi=0$ as the field equation defining a coupled variant of $P$  with `in' and `out' regions given by $M^\pm = I^\pm(\Sigma^\pm)$, where $I^+(N)$ and $I^-(N)$ are the chronological future and past respectively of a set $N$. Clearly the coupling region is only timelike compact, rather than compact (unless $M$ has compact Cauchy surfaces) -- for this 
reason this model does not fully conform to the FV framework as originally formulated. Nonetheless, it is instructive to pursue it here, because it provides particularly simple results.

It may easily be shown that the scattering operator is given by $\Theta\Phi(F)=- \mathrm{i}\Phi(F)$, where $\Phi(F)$ is a smearing of the complex field against
$F\in C_c^\infty(M;\mathbb{C})$, see Appendix~\ref{sec_appendix_proof_of_p}. From the perspective of the complex field, $\Theta$ is nothing but a particular global $U(1)$ gauge transformation. However, the complex field was introduced only as a convenient technical device to combine two observable real scalar fields, and at this level, the $U(1)$ action is not a gauge symmetry but acts nontrivially on observables; the $U(1)$ symmetry is broken. Indeed, the real scalar fields of the system and probe are transformed by
\begin{equation}
    \Theta \varphi_\mathcal{S}(f) = -\varphi_\mathcal{P}(f), \qquad \Theta\varphi_\mathcal{P}(f) =\varphi_\mathcal{S}(f) 
\end{equation}
for all $f\in C_c^\infty(M;\mathbb{R})$, and, at the level of exponentiated fields $W_\mathcal{S}(f)= e^{\mathrm{i} \varphi_\mathcal{S}(f)}$ and $W_\mathcal{P}(f)= e^{\mathrm{i} \varphi_\mathcal{P}(f)}$, one has
\begin{equation}
    \Theta W_\mathcal{S}(f)\otimes W_\mathcal{P}(h) = W_S(h)\otimes W_\mathcal{P}(-f).
\end{equation}
It is now easy to read off the induced observable map in this model. Specifically, Eq.~\eqref{eq_induced_obs} gives
\begin{equation}\label{eq:epspp}
  \varepsilon_{\sigma}^{\chi}(W_\mathcal{P}(h)) = W_\mathcal{S}(h)
\end{equation}
(the superscript $\chi$ serves to indicate the coupling model chosen) for all $h\in C_c^\infty(M;\mathbb{R})$, and also that
\begin{equation}
    \varepsilon_{\sigma}^{\chi}(
\varphi_\mathcal{P}(f_1)\cdots \varphi_\mathcal{P}(f_n)
    ) =\varphi_\mathcal{S}(f_1)\cdots \varphi_\mathcal{S}(f_n).
\end{equation}
for all $f_1,\ldots,f_n\in C_c^\infty(M;\mathbb{R})$. From these
formulae it follows easily that $\varepsilon_{\sigma}^{\chi}$ induces isomorphisms between the probe and system theories, whether these
are quantized as Weyl algebras (\Sect\ref{sec:Weyl_algebra}) or in terms of smeared fields (\Sect\ref{sec:field_algebra}).

Summarising, we have given an explicit model in which every local observable can be measured via a probe. However, the model is not entirely satisfactory, because it requires (in general) a noncompact coupling zone. Nonetheless, as a proof of concept, it may be taken as an indication that the set of inducible observables forms a large subset of the local system observables. 

In the following we will show that this is indeed the case, even reimposing the compactness of the coupling zone.

\section{Classical asymptotic measurement schemes}\label{sec:classical}

The proofs of our main results for measurement schemes in QFT make essential use
of analogous results for classical scalar fields, which will be set out in this section.

\subsection{Concepts from Lorentzian geometry}

For the convenience of the reader and in order to fix notation, we collect some standard properties of globally hyperbolic spacetimes. Our signature convention is mostly minus, i.e., $(+,-,\dots,-)$. A $1+d$-dimensional spacetime $M$, i.e., a smooth Lorentzian oriented and time-oriented manifold with finitely many connected components, is globally hyperbolic if and only if it contains a Cauchy hypersurface. In this case, it is isometric to $\mathbb{R}\times \Sigma$ with metric $\Omega^2(\mathrm{d}\tau^2 \oplus- h_\tau)$, where $\Omega$ is smooth, nowhere vanishing and $\tau \mapsto h_\tau$ is a smooth family of Riemannian metrics of the $d$-dimensional manifold $\Sigma$, and each $\{\tau\} \times \Sigma$ is a smooth Cauchy surface, see~\cite{bernal2003smooth}. For a subset $N \subseteq M$, we denote by $J^+(N)$ and $J^-(N)$ its causal future and past respectively and by $D(N)$ its domain of dependence or Cauchy development. $N$ is called causally convex if and only if $N = J^+(N) \cap J^-(N)$. Non-empty subsets of a globally hyperbolic spacetime $M$ that are open, causally convex and have finitely many connected components will be called \emph{regions}; when equipped with the inherited metric, orientation and causal structure from $M$, they become globally hyperbolic spacetimes in their own right. 
The causal complement of a subset $K$ is $M\setminus \qty(J^+(K) \cup J^-(K))$, and two regions are described as spacelike separated if one lies in the causal complement of the other. For compact $K$, the sets $M\setminus J^\pm(K)$ are regions; see, for instance, the Appendix of~\cite{Fewster_2012} for details and proofs.\\

\subsection{Systems of linear scalar fields}

Consider a collection of $k$ real scalar classical fields $\Phi:= (\varphi_1, ..., \varphi_k)^T$ on a globally hyperbolic spacetime $M$ satisfying a \emph{linear} second order normally hyperbolic partial differential equation of motion $P \Phi=0$~\cite{bar2007wave, Baer:2015},
which is formally self-adjoint, i.e., 
\begin{equation}
    \int_M (P\Phi)\cdot \Psi\,\mathrm{d}V_M =  \int_M  \Phi\cdot (P\Psi)\,\mathrm{d}V_M  
\end{equation}
for any $\Phi,\Psi\in C^\infty(M;\mathbb{R}^k)$ with compactly intersecting supports, and where the dot denotes the standard inner product in $\mathbb{R}^k$. The general form of such an operator is
\begin{equation}
    P\Phi = \Box \Phi + V^\alpha\nabla_\alpha\Phi + W\Phi
\end{equation}
where $V^\alpha$ and $W$ are smooth matrix-valued coefficients,
with $V^\alpha$ antisymmetric and $W-W^T=\nabla_\alpha V^\alpha$, and the equation $P\Phi=0$ is the Euler--Lagrange equation of the Lagrangian density
\begin{equation}
    \mathcal{L}= \tfrac{1}{2}\sqrt{-g} (\nabla^\alpha \Phi\cdot \nabla_\alpha \Phi - \Phi \cdot V^\alpha\nabla_\alpha\Phi - \Phi \cdot W\Phi).
\end{equation}

A simple example is a pair of independent Klein-Gordon fields with masses $m_1, m_2  \geq 0$, for which $P=(\Box + m_1^2) \oplus (\Box + m_2^2)$; another example is the operator $Q$ defined in~\eqref{eq:ppQ}, understood as an operator on $C^\infty(M;\mathbb{R}^2)$.
 Associated to a normally hyperbolic equation on $M$ are unique advanced ($-$) and retarded ($+$) Green operators $E_P^\pm$, whose difference defines $E_P:=E_P^- - E_P^+$. It holds in particular that $\mathrm{supp}\; E^\pm_P f \subseteq J^\pm (\mathrm{supp} f)$ for all $f\in C_c^\infty(M;\mathbb{R}^k)$.

The real vector space of real-valued solutions to the equations of motion with spatially compact ($sc$) support,
\begin{equation}
    \mathrm{Sol}_{sc}(P):=\{\Phi\in C_{sc}^\infty(M;\mathbb{R}^k): P\Phi =0\}
\end{equation}
is then isomorphic to the quotient $C_c^\infty(M;\mathbb{R}^k) / PC_c^\infty(M;\mathbb{R}^k)$ via
\begin{equation}\label{eq:Sol_iso}
    \begin{aligned}
    E_P :\;   C_c^\infty(M;\mathbb{R}^k) / PC_c^\infty(M;\mathbb{R}^k) &\to \mathrm{Sol}_{sc}(P)\\
     [f]_P &\mapsto E_P f,
    \end{aligned}
\end{equation}
where $[f]_P= f + PC_c^\infty(M;\mathbb{R}^k)$ denotes the equivalence class of $f$ in $C_c^\infty(M;\mathbb{R}^k) / PC_c^\infty(M;\mathbb{R}^k)$. We will describe test functions as being equivalent if they belong to a common equivalence class in this sense. Note that the map in~\eqref{eq:Sol_iso} is well-defined because $ E_P P f =0$. A fact of fundamental importance (the classical timeslice property) is that if a region $N$ contains a Cauchy surface for $M$ then every equivalence class $[f]_P$
has a representative supported in $N$: $C_c^\infty(M;\mathbb{R})=
C_c^\infty(N;\mathbb{R})+ P C_c^\infty(M;\mathbb{R})$.
The space $C_c^\infty(M;\mathbb{R}^k) / PC_c^\infty(M;\mathbb{R}^k)$ becomes a symplectic space once equipped with the symplectic form
\begin{equation}\sigma_P([f],[g])=E_P(f,g):= \int_M f \cdot (E_P g) \; \mathrm{d}V_M.
\end{equation}
It also carries a natural quotient topology ${ \tau^{\mathrm{cl}}}$ obtained from the standard test function topology on $C_c^\infty(M;\mathbb{R}^k)$, which is the only
topology we will consider on $C_c^\infty(M;\mathbb{R}^k)$.

To maintain the analogy with~\cite{fewster2018quantum}, we define the \emph{classical} theory as a net of $\mathbb{R}$-vector spaces $\mathcal{C}_\mathcal{P}$ by setting
\begin{equation}
        \mathcal{C}_\mathcal{P}(N):= C_c^\infty(N;\mathbb{R}^k)/P C_c^\infty(N;\mathbb{R}^k)
\end{equation}
for each region $N$ of $M$, regarding $\mathcal{C}_\mathcal{P}(N)$ as a subspace of
$\mathcal{C}_\mathcal{P} :=\mathcal{C}_\mathcal{P}(M)=C_c^\infty(M;\mathbb{R}^k)/P C_c^\infty(M;\mathbb{R}^k)$, with the inherited symplectic form and topology.\footnote{Our choice to suppress the dependence of $\mathcal{C}_\mathcal{P}$ on the fixed spacetime $M$ causes an overload of notation, since $\mathcal{C}_\mathcal{P}(N)$ could either be the global space on the fixed spacetime $N$, or a subspace associated to the region $N$ inside a bigger fixed spacetime. However, it is easy to see that these two structures are in fact isomorphic.}

States of this theory are given by real-valued distributional solutions to the equations of motion, i.e., $\kappa \in \mathscr{D}'(M;\mathbb{R}):= \qty(C_c^\infty(M;\mathbb{R}))'$ such that $\kappa$ vanishes on $P C_c^\infty(N;\mathbb{R})$. Note that there is a distinguished state given by the zero solution, and that all states are continuous with respect to the quotient topology. Moreover, 
there are enough states to separate the observables:
if $\kappa([f]_{ P}) = \kappa([g]_{ P})$ for all
$\kappa$ of the form ${\kappa([f]_{P})=} \int f (E_{ P} h) \mathrm{d}V_M$ for some $h \in C_c^\infty(M;\mathbb{R})$, then the fundamental lemma of variational calculus implies that $f-g \in \mathrm{ker}(E_{ P}) = \mathrm{im}({ P})$ and hence $\qty[f]_{P} = \qty[g]_{ P}$.

\subsection{Combinations of systems and probes}\label{sec:classical_combinations}

Let us now assume we have a single linear real system field $\varphi_\mathcal{S}$ with equation of motion operator $S$ on $C^\infty(M;\mathbb{R})$ and $k$ linear real probe fields $\qty(\varphi_\mathcal{P})_j$ with equation of motion operator $P$ on $C^\infty(M;\mathbb{R}^{k})$ both understood here (and throughout this paper) to be linear, normally hyperbolic and formally self-adjoint, with Lagrangian densities
$\mathcal{L_S}$ and $\mathcal{L_P}$ respectively.

The coupling between the system and probe will be a bilinear coupling that (in order to fit into the FV framework) is only active in a compact coupling zone $K\subset M$. An example for such a coupling can easily be written down in the language of Lagrangian densities:
\begin{equation}
    \begin{aligned}
    \mathcal{L}= \mathcal{L_S} + \mathcal{L_P} - {\lambda} \sum\limits_{j=1}^k  \rho_j \varphi_\mathcal{S} \qty(\varphi_\mathcal{P})_j,
    \end{aligned}
\end{equation}
where $\rho_1,...,\rho_j \in C_c^\infty(M;\mathbb{R})$ are coupling functions with support in $K$ and $\lambda$ is a common coupling constant. For every $\lambda \in \mathbb{R}$, this gives rise to the coupled equation of motion $T{_\lambda}$ on $C^\infty(M;\mathbb{R}^{k+1}) \simeq C^\infty(M;\mathbb{R}) \oplus C^\infty(M;\mathbb{R}^{k})$ conveniently defined in block matrix notation by
\begin{equation}
    \begin{aligned}
    T{_\lambda}:= \mqty({S} & {\lambda} R^T \\ {\lambda} R& {P}),
    \end{aligned}
    \label{eq_T_lambda}
\end{equation}
where $S:C^\infty(M;\mathbb{R}) \to C^\infty(M;\mathbb{R})$ is a $1 \times 1$ matrix, $P: C^\infty(M;\mathbb{R}^k) \to C^\infty(M;\mathbb{R}^k)$ is a $k \times k$ matrix, and $R$ and $R^T$ are $k\times 1$ and $1\times k$ matrices so that
\begin{equation}
    \begin{aligned}
    R= \mqty(R_1 \\ \vdots \\ R_k) : C^\infty(M;\mathbb{R}) \to C^\infty(M;\mathbb{R}^k); \qquad &f \mapsto \mqty(R_1 f \\ \vdots \\ R_k f),\\
    R^T= (R_1, \dots, R_k) : C^\infty(M;\mathbb{R}^k) \to C^\infty(M;\mathbb{R}); \qquad &\vec{f} \mapsto \sum\limits_{j=1}^k R_j f_j.
    \label{eq_def_R}
    \end{aligned}
\end{equation}
Here, $R_j$ is the operator of point-wise multiplication with $\rho_j$. For every $\lambda \in \mathbb{R}$, $T{_\lambda}$ is a formally self-adjoint normally hyperbolic operator and therefore gives rise to 
a well-defined classical theory; for $\lambda=0$ this is the uncoupled combination
$T_0=S\oplus P$, leading to classical theory $\mathcal{C}_{\mathcal{S}}\oplus {\mathcal{C}_\mathcal{P}}$, while it describes a coupled theory for any $\lambda\neq 0$.

The scattering map $\vartheta_\lambda:\mathcal{C}_{\mathcal{S}}(M)\oplus \mathcal{C}_{\mathcal{P}}(M)
\to\mathcal{C}_{\mathcal{S}}(M)\oplus \mathcal{C}_{\mathcal{P}}(M)$ describing the dynamics of $T_\lambda$ relative to that of $T_0$ may be read off from formulae in~\cite{fewster2018quantum}. 
Let us define the `out' region $M^+:= M \setminus J^-(K)$ and the `in' region $M^-:= M \setminus J^{+}(K)$. For any
$ \qty[F]_{S\oplus P}$ with representative $F \in C_c^\infty(M^+;\mathbb{R}^{k+1})$, we have
\begin{equation}\label{eq:classicalscattering}
    \vartheta_\lambda \qty[F]_{S\oplus P} = \qty[\tilde{F}]_{S\oplus P}
\end{equation}
where $\tilde{F}\in C_c^\infty(M^-;\mathbb{R}^{k+1})$ is any test function with the property that 
\begin{equation}\label{eq:scatteringmapeqn}
    E_{T{_\lambda}} \tilde{F} = E_{T{_\lambda}} F,
\end{equation}
i.e., $\tilde{F}$ generates the same solution to the classical homogeneous coupled field equation as $F$. Note that it follows from the equation of motion that for any region $U\subseteq M^-$ satisfying $\supp F\subseteq D(U)$, there exists such an $\tilde{F}$ supported in $U$.

Given the support of $\tilde{F}$, and the fact that $T_{\lambda}$ and $S \oplus P$ agree in $M^-$, condition~\eqref{eq:scatteringmapeqn}
fixes $\tilde{F}$ modulo the possible addition of terms of the form $(S \oplus P) H$ for $H\in C_c^\infty(M^-;\mathbb{R}^{k+1})$, which do not change the left-hand side of~\eqref{eq:scatteringmapeqn}. Furthermore, the right-hand side of~\eqref{eq:classicalscattering} is unchanged if $\tilde{F}$ is modified by terms of the form {$(S \oplus P) H$} with $H$ compactly supported anywhere in $M$, and this freedom can be exploited to find convenient formulae. For instance, it may be shown that
\begin{equation}\label{eq:convenientclassicaltheta}
\vartheta_\lambda \qty[F]_{S\oplus P} =\qty[F-(T_\lambda - S\oplus P)E_{T_\lambda}^-F]_{S \oplus P}.
\end{equation}
For a derivation, see Appendix D of \cite{fewster2018quantum}; note that while $F- (T{_\lambda}-{S\oplus P})E_{T{_\lambda}}^- F$ is not supported in $M^-$, it differs from a function that is by a term {$(S\oplus P)H$}, for some $H \in C_c^\infty(M;\mathbb{R}^{k+1})$. 
It will also be convenient to write
\begin{equation}\label{eq:unvartheta_def}
    \theta_\lambda F = F-(T_\lambda - S\oplus P)E_{T_\lambda}^-F
\end{equation}
so that $\vartheta_\lambda[F]_{S\oplus P}=[\theta_\lambda F]_{S\oplus P}$ for $F \in C_c^\infty(M^+;\mathbb{R}^{k+1})$, as in Eq.~(29) in~\cite{Ruep2021}.

\subsection{Induced observables and classical asymptotic measurement schemes}

Induced observables may be introduced by analogy with~\eqref{eq_induced_obs}. 
It is convenient to write a general test function $F\in C_c^\infty(M;\mathbb{R}^{k+1})$
as $F=f\oplus \vec{g}$ where $f\in C_c^\infty(M;\mathbb{R})$ and $\vec{g}\in C_c^\infty(M;\mathbb{R}^{k})$. Then we seek a map $\varepsilon^{\mathrm{cl}, \lambda R}:\mathcal{C}_{\mathcal{P}}\to \mathcal{C}_{\mathcal{S}}$
such that 
\begin{equation}
    \kappa (\varepsilon^{\mathrm{cl}, \lambda R}([\vec{h}]_{P}))= (\kappa \oplus 0) \qty(\vartheta_\lambda  [0 \oplus \vec{h}]_{S\oplus P}) 
\end{equation}
holds for all $[\vec{h}]_{P}\in \mathcal{C}_{\mathcal{P}}$ and all states $\kappa$ on $\mathcal{C}_{\mathcal{S}}(M)$, which is the analogue
of~\eqref{eq_induced_obs} using the distinguished zero solution as the probe preparation state. Because the states separate the observables, there is a unique solution, namely
\begin{equation}
    \varepsilon^{\mathrm{cl}, \lambda R}([\vec{h}]_{P}) := \text{pr}_1 (\vartheta_\lambda [0 \oplus \vec{h}]_{S\oplus P}),
\end{equation}
where $\text{pr}_1: \mathcal{C}_\mathcal{S} \oplus \mathcal{C}_\mathcal{P} \to \mathcal{C}_\mathcal{S}$ is the projection on the first component,  $\text{pr}_1(A \oplus B):= A$. 

Now allow $\vec{h}$ to vary with $\lambda$. Writing
\begin{equation}\label{eq:convenient}
\begin{aligned}
    \mqty(f_\lambda \\ \vec{g}_\lambda)&=\theta_\lambda \mqty(0\\ \vec{h}_\lambda)
    =\mqty(0\\ \vec{h}_\lambda) -(T_\lambda - S \oplus P)E^-_{T_\lambda} \mqty(0\\ \vec{h}_\lambda)\\
    &= \mqty(0\\ \vec{h}_\lambda) -\mqty(0 & \lambda R^T \\ \lambda R & 0)E^-_{T_\lambda} \mqty(0\\ \vec{h}_\lambda),
    \end{aligned}
\end{equation}
and comparing with~\eqref{eq:convenientclassicaltheta}, one has
\begin{equation}
    \varepsilon^{\mathrm{cl}, \lambda R}([\vec{h}_\lambda]_P) = \qty[f_\lambda]_S.
    \label{eq_classical_varepsilon}
\end{equation}
In particular, because $f_\lambda\in \Ran R^T$ is supported in $\supp R$, we immediately see that $ \varepsilon^{\mathrm{cl}, \lambda R}([\vec{h}_\lambda]_P)$ can be localised in any region containing the support of $R$.

In this way $(\mathcal{C}_{\mathcal{P}},\varepsilon^{\mathrm{cl}, \lambda R},[\vec{h}_\lambda]_P)$ forms a measurement scheme for $[f_\lambda]_S\in \mathcal{C}_{\mathcal{S}}$ for each $\lambda>0$. Turning to asymptotic measurement schemes, we will prove:

\begin{theo}
For every precompact region $N$, every $\qty[f]_S \in \mathcal{C}_\mathcal{S}(N)$ and every region $L \subseteq M \setminus J^-(\overline{N})$ such that $N \subseteq D(L)$, there exists a ${ \tau^{\mathrm{cl}}}$-\emph{asymptotic measurement scheme} for $f$ with coupling zone in $N$ and \emph{processing region} $L$. 
\label{theo_classical_asympt_m_scheme}
\end{theo}

\paragraph*{Remark:} In fact, for an admissible processing region $L$, i.e., $L \subseteq M \setminus J^-(\overline{N})$ it holds that $N \subseteq D(L)$ if and only if $N \subseteq D^-(L)$.\footnote{If $p\in N \cap D^+(L)$, then every past inextendible causal curve through $p$
must intersect $L$, so $L \cap J^-(\overline{N}) \neq \emptyset$, which is a contradiction.}\\

In fact the asymptotic measurement scheme we will construct uses the $k=1$ probe theory only {-- theories with $k\ge 1$ will play a role later}. The construction proceeds in two steps.

The first step is to find a subset $\tilde{N}\subseteq N$ for which the following exist:
\begin{itemize}
    \item an equivalent test function $\tilde{f}\in[f]_S$  supported in $\tilde{N}$
    \item a real solution $\varphi$ of the probe field equation so that $\varphi$ is nonvanishing on $\tilde{N}$;
    \item a real-valued test function $h$ supported in $L$ so that $\varphi=E_P h$.
\end{itemize}
Consequently, there is a unique $\rho\in C_c^\infty(\tilde{N};\mathbb{R})$ defined by $\rho = - \frac{\tilde{f}}{\varphi}$, and associated pointwise multiplication operator $R$. Note that multiplication by $\varphi^{-1}$ is well defined here as $\varphi$ is non-zero on the support of $\tilde{f}$.

In the second step we show that for the coupling $\lambda R$, and the probe observables labelled by $h_\lambda=h/\lambda$, one has $\varepsilon^{\mathrm{cl}, \lambda R}(\qty[h_\lambda]_P) \to 
\qty[\tilde{f}]_S=
\qty[f]_S$ in $\qty(\mathcal{C}_\mathcal{S}, { \tau^{\mathrm{cl}}})$, from which the asymptotic measurement scheme is easily constructed.

The first step is accomplished by the following Lemma, which is a consequence of Lemma~\ref{lem_finding_slab} and proved in Appendix~\ref{sec_appendix_proof_theo_reduction_slab}.

\begin{lem}
Let $S, P$ be the system and probe equation of motion operators. For every region $N \subseteq M$ and every test function $f \in C_c^\infty(N;\mathbb{R})$ there exists a precompact region $\tilde{N} \subseteq N$ and $\tilde{f} \in C_c^\infty(\tilde{N};\mathbb{R})$ such that
\begin{enumerate}
    \item $ [f]_S= [\tilde{f}]_S$,
    \item $ \exists \, \varphi \in \mathrm{Sol}_{sc}(P)$ such that $\varphi \restriction \tilde{N}$ is nowhere vanishing.
\end{enumerate}
Moreover, for any region $L\subseteq M^+:=M\setminus J^-( \mathrm{supp} \tilde{f})$ whose domain of dependence contains $N$ there exist $h \in C_c^\infty(L;\mathbb{R})$ and $\rho \in C_c^\infty(\tilde{N};\mathbb{R})$ such that
\begin{equation}
    \tilde{f}=- R E_{P}^- h,
\end{equation}
where $R$ is the operator of pointwise multiplication with $\rho$. In fact, the function $\rho$ is independent of the choice of $L$.
\label{lem_reduction_to_slab_classical}
\end{lem}

The second step above is the content of the following lemma, which is a special case of Lemma~\ref{lem_approx_single_test_multiple_probe} proved below.

\begin{lem}
For $\rho \in C_c^\infty(M;\mathbb{R})$ let $R: C^\infty(M;\mathbb{R}) \to C^\infty(M;\mathbb{R})$ be the operator of pointwise multiplication with $\rho$ and let $P$ be the probe equation of motion operator. For $h \in C_c^\infty(M;\mathbb{R})$ and $\lambda >0$ define $h_\lambda := h/\lambda$. Then
\begin{equation}
    \lim_{\lambda \to 0} \varepsilon^{\mathrm{cl}, \lambda R}([h_\lambda]_{P}) = \qty[-R E_{P}^-h]_S
\end{equation}
in $\qty(\mathcal{C}_\mathcal{S}, {\tau^{\mathrm{cl}}})$.
\label{lem_approx_single_test_multiple_probe_classical}
\end{lem}

\paragraph*{Remark:} The fact that $h_\lambda$ diverges as $\lambda\to 0$ is in fact unproblematic since there is no convergence requirement on the probe observables of an asymptotic measurement scheme. See however \Sect\ref{sec:effort} for the effect it has on the ``effort'' required to measure to finer accuracy. A further discussion is also given in \Sect\ref{sec_variance}.\\

We are now ready to put things together and prove Theorem~\ref{theo_classical_asympt_m_scheme}.

\begin{proof}[Proof of Theorem~\ref{theo_classical_asympt_m_scheme}]
Let $\qty[f]_S \in \mathcal{C}_\mathcal{S}(N)$ for a precompact region $N$ and let $L$ be a region in $ M \setminus J^-(\overline{N})$ such that $N \subseteq D(L)$. Then, according to Lemma~\ref{lem_reduction_to_slab_classical}, $\qty[f]_S = \qty[-RE_P^-h]_S$ for some $h \in C_c^\infty(L;\mathbb{R})$ and  $\rho \in C_c^\infty(\tilde{N};\mathbb{R})$, for some pre-compact region $\tilde{N}\subseteq N$, and according to Lemma~\ref{lem_approx_single_test_multiple_probe_classical} 
\begin{equation}
    \lim\limits_{\lambda \to 0} \varepsilon^{\mathrm{cl}, \lambda R}(\qty[h_\lambda]_P) = \qty[-RE_P^-h]_{S} =\qty[f]_S, 
\end{equation}
where $h_\lambda = h/\lambda$. In particular, the collection 
\begin{equation}
    H^\mathrm{cl}_\lambda = \qty(\mathcal{C}_\mathcal{P}, \varepsilon^{\mathrm{cl}, \lambda R}, \qty[h_\lambda]_P),
    \qquad (\lambda>0),
    \label{eq_classical_ams}
\end{equation}
forms a $\tau^\mathrm{cl}$-asymptotic measurement scheme for $\qty[f]_S$ as $\lambda\to 0^+$ with coupling in $N$ and processing region $L$.
\end{proof}

\subsection{Effort and rate of convergence}\label{sec:effort}

We now discuss the efficiency of asymptotic measurement schemes by comparing the rate of convergence with a measure of the effort required. We will focus on the classical measurement schemes
\begin{equation}
    H^\mathrm{cl}_\lambda = \qty(\mathcal{C}_\mathcal{P}, \varepsilon^{\mathrm{cl}, \lambda R}, \qty[h_\lambda]_P),
\end{equation}
as constructed in Theorem~\ref{theo_classical_asympt_m_scheme}. To quantify the effort associated with a measurement scheme, let $\mathrm{eff}:\mathcal{C}_\mathcal{P} \to \mathbb{R}^+_0$ be some arbitrary but fixed choice of seminorm on $\mathcal{C}_\mathcal{P}$ such that $\mathrm{eff}(\qty[h]_P) \neq 0$. In \Sect\ref{sec_variance}, where we discuss the physical interpretation of these asymptotic measurement schemes in more detail, we give one physical example of a seminorm that encodes experimental effort for quantum asymptotic measurement schemes. This concept is very general, however, and a seminorm can in principle encode a multitude of different experimental factors relating to the feasibility of a measurement.

Equipped with some seminorm, $\mathrm{eff}$, the effort associated with a measurement of the probe observable $\qty[h_\lambda]_P$ 
diverges as $\lambda\to 0^+$,
\begin{equation}\label{eq_eff_diverges_for_probe}
\mathrm{eff}\qty(\qty[h_\lambda]_P) = \lambda^{-1} \mathrm{eff}\qty(\qty[h]_P) \to \infty .
\end{equation}
However, the reward for this effort is that the induced observables approach the desired limit $\varepsilon^{\mathrm{cl}, \lambda R}(\qty[h_\lambda]_P)\to \qty[f]_S$ as $\lambda\to 0^+$. The rate at which convergence occurs, relative to the effort involved, can be used as a measure of the efficiency of the asymptotic measurement scheme. 
In general, we will say that the scheme is of order $n$ if 
\begin{equation}\label{eq_def_order_of_m_scheme} 
    \varepsilon^{\mathrm{cl}, \lambda R}(\qty[h_\lambda]_P) - \qty[f]_S = \mathrm{eff}\qty(\qty[h_\lambda]_P)^{-n} \; \mathcal{E}(\lambda),
\end{equation}
for some $\lambda \mapsto \mathcal{E}(\lambda) \in \mathcal{C}_\mathcal{S}$ that is bounded with respect to the topology $\tau^{\mathrm{cl}}$ as $\lambda\to 0^+$. The required boundedness certainly holds if $\mathcal{E}(\lambda)$ extends continuously to a neighbourhood of $\lambda=0$ since then there exists $\epsilon >0$ such that $\lbrace \mathcal{E}(\lambda)|\lambda\in [0,\epsilon] \rbrace$ is compact (as it is the image of a compact set under a continuous map) and hence bounded. Note that this definition is invariant under reparametrisation of $\lambda$ and is independent of the choice of $\mathrm{eff}$, provided that $\mathrm{eff}(\qty[h]_P) \neq 0$. 

The classical asymptotic measurement scheme constructed in Theorem~\ref{theo_classical_asympt_m_scheme} turns out to be second order according to the above definition. It also turns out that we can improve the situation by increasing $k$, the number of probe fields. This is the statement of the following theorem, whose proof is {presented} in Appendix~\ref{sec_appendix_faster_approx}.

\begin{theo}
For every precompact region $N$, every admissible processing region $L$, every $k$ and every $\qty[f]_S \in \mathcal{C}_\mathcal{S}(N)$ there is an asymptotic measurement scheme for $\qty[f]_S$ of order $2k$.
\label{theo_faster_approx}
\end{theo}

\subsection{Combining measurement schemes}\label{sec_comb_measurement_schemes}

Suppose that one has asymptotic measurement schemes for two or more classical observables.
Each involves a specific sequence of couplings. How can these be combined to find an asymptotic measurement scheme for their sum? Considering situations where the supports of these coupling functions overlap, it is clear that one cannot simply add the coupling functions in general. A better solution is to take the direct sum of all the probe systems, and couple each to the system field as before. Even this is not a trivial matter, because the various probe fields now interact with each other via their coupling to the system. Nonetheless, because this coupling is at a higher order than the direct coupling of each probe to the system, one may prove the following. For simplicity of notation,
we restrict to the situation in which each probe field has the same free equation of motion operator $P$, and denote the equation of motion operator for $l$ probes
by $P^{\oplus l}$.

\begin{theo}
For $j=1,2,...,l$, let
\begin{equation}
    H^{\mathrm{cl},j}_\lambda = \qty(\mathcal{C}_\mathcal{P}, \varepsilon^{\mathrm{cl}, \lambda R_j}, [h^j_\lambda]_P),
\end{equation}
be classical asymptotic measurement schemes with coupling in $N$ and processing region $L$ for the observables $\qty[f_j]_S  = \lim_{\lambda\to 0} \varepsilon^{\mathrm{cl}, \lambda R_j}([ h^j_\lambda ]_{P})= \qty[-R_jE_P^- h^j]_S \in \mathcal{C}_\mathcal{S}(N)$. Then 
\begin{equation}
    H^{\mathrm{cl}}_\lambda := \qty(\mathcal{C}_\mathcal{P}^{\oplus l}, \varepsilon^{\mathrm{cl}, \lambda R }, [\vec{h}_\lambda]_{P^{\oplus l}}),
\end{equation}
where $R=(R_1 , ... ,R_l)^T$ as in Eq.~\eqref{eq_def_R}, $P^{\oplus l}:= P \oplus ... \oplus P$ and $\vec{h}_\lambda := \bigoplus_{j=1}^l h_\lambda^j$, is an asymptotic measurement scheme for $\sum_{j=1}^l [f_j ]_S =\qty[-R^T E_{P^{\oplus l}}^- \vec{h}]_S$.
\label{theo_comb_classical_mscheme}
\end{theo}

This result is an immediate consequence of the following lemma (of which Lemma~\ref{lem_approx_single_test_multiple_probe_classical} is a special case).

\begin{lem}
For $\rho_1, ..., \rho_k \in C_c^\infty(M;\mathbb{R})$ let $R^T: C^\infty(M;\mathbb{R}^k) \to C^\infty(M;\mathbb{R})$ be the operator defined in Eq.~\eqref{eq_def_R} and let $P^{\oplus k}= P \oplus ... \oplus P$. For $\vec{h}=(h^1, ..., h^k)^T \in C_c^\infty(M;\mathbb{R}^k)$ and $\lambda >0$ define $\vec{h}_\lambda := \vec{h}/\lambda$ and $h^j_\lambda:= h^j/\lambda$. Then
\begin{equation}
    \lim_{\lambda \to 0} \varepsilon^{\mathrm{cl}, \lambda R}([\vec{h}_\lambda]_{P^{\oplus k}}) =  \lim_{\lambda \to 0} \sum\limits_{j=1}^k \varepsilon^{\mathrm{cl}, \lambda R_j}([ h^j_\lambda ]_{P}) = \qty[-R^T E_{P^{\oplus k}}^- \vec{h}]_S
\end{equation}
in $\qty(\mathcal{C}_\mathcal{S}, {\tau^{\mathrm{cl}}})$.
\label{lem_approx_single_test_multiple_probe}
\end{lem}

\paragraph*{Remark:}  For $\lambda>0$, $\varepsilon^{\mathrm{cl}, \lambda R}([\vec{h}_\lambda]_{P^{\oplus k}})$ and $\sum\limits_{j=1}^k \varepsilon^{\mathrm{cl}, \lambda R_j}([ h^j_\lambda ]_{P})$ are generally unequal because the probe fields interact with each other via the system.\\

\begin{proof}
Recall from Eq.~\eqref{eq:convenient} and Eq.~\eqref{eq_classical_varepsilon} that $\varepsilon^{\mathrm{cl}, \lambda R}([\vec{h}_\lambda]_{P^{\oplus k}}) = \qty[f_\lambda]_S$, where
\begin{equation}\label{eq:flambdaglambda}
\begin{aligned}
    \mqty(f_\lambda \\ \vec{g}_\lambda):&= \mqty(0\\ \vec{h}_\lambda) -\mqty(0 & \lambda R^T \\ \lambda R & 0)E^-_{T_\lambda} \mqty(0\\ \vec{h}_\lambda)\\
    &=\mqty(0\\ \vec{h}_\lambda) -\mqty(0 & R^T \\  R & 0)E^-_{T_\lambda} \mqty(0\\ \vec{h}).
    \end{aligned}
\end{equation}

Let us extend $T_\lambda$ and its Green operators to spaces of complex-valued functions in the obvious way, also allowing $\lambda$ to be complex. It is shown by one of us (CJF) in~\cite{fewster-non-local} that $\lambda \mapsto E_{T_\lambda}^-$ is holomorphic on $\mathbb{C}$ with respect to the topology of bounded convergence of continuous linear maps from the LF space $C_c^\infty(M;{\mathbb{C}^{k+1}})$ to the Fr\'{e}chet space $C^\infty(M;{\mathbb{C}^{k+1}})$. (See~\cite{Treves:TVS} for the definition of the topologies involved.) As the linear operator $\mqty(0 & R^T \\  R & 0)$ {is continuous} from $C^\infty(M;{\mathbb{C}^{k+1}})$ to $C_c^\infty(M;{\mathbb{C}^{k+1}})$, the map
\begin{equation}
    \lambda \mapsto \mqty(0 & R^T \\  R & 0)E^-_{T_\lambda}
\end{equation}
is holomorphic on all of $\mathbb{C}$ with respect to the topology of bounded convergence of continuous linear maps from $C_c^\infty(M;{\mathbb{C}^{k+1}})$ into itself~\cite{fewster-non-local}. In particular,
\begin{equation}
    \mqty(0 & R^T \\  R & 0)E^-_{T_\lambda} \mqty(0\\ \vec{h}) \overset{\lambda \to 0}{\longrightarrow} \mqty(0 & R^T \\  R & 0)\mqty(E_S^- & 0^T \\  0 & E_{P^{\oplus k}}^-)\mqty(0\\ \vec{h}) = \mqty(R^T E^-_{P^{\oplus k}} \vec{h}\\\vec{0})
\end{equation}
in $C_c^\infty(M;{\mathbb{C}^{k+1}})$ and hence $f_\lambda \to -R^T E^-_{P^{\oplus k}} \vec{h} = -\sum_{j=1}^k R_j E_P^- h^j$ in $C_c^\infty(M;{\mathbb{R}})$. Finally, the result follows by continuity of the quotient map.
\end{proof}

\paragraph*{Remark:} 
Analyticity of $\lambda\mapsto E_{T_\lambda}^-$ implies in particular that
\begin{equation}
E_{T_\lambda}^- = E_{T_0}^- -\lambda  E_{T_0}^- \mqty( 0 & R^T\\ R & 0)E_{T_0}^- +   \mathcal{O}(\lambda^2). 
\end{equation}
Substituting in Eq.~\eqref{eq:flambdaglambda}, one obtains 
\begin{equation}\label{eq:glambda_expansion}
    \vec{g}_\lambda = \lambda^{-1} \vec{h}  + \lambda RE_S^- R^T E^-_{ P^{\oplus k}}   \vec{h}+ \mathcal{O}(\lambda^3),
\end{equation}
which will be needed later (see also Eq.~(5.27) in~\cite{fewster2018quantum}). Here, the diagonal structure of $E_{T_0}^-$ 
and the off-diagonal nature of the interaction term combine to eliminate even powers from the expansion. 
\\

The above lemma will be an important ingredient in the discussion of quantum asymptotic measurement schemes below. It is nonetheless interesting to consider whether asymptotic measurement schemes may be combined more abstractly and some thoughts in that direction are collected in Appendix~\ref{sec_appendix_abstract comb}.

\section{Quantum asymptotic measurement schemes: the field algebra}\label{sec:asymp_m_schemes_field_algebra}

\subsection{The field algebra}\label{sec:field_algebra}

Consider a classical equation of motion operator $P$ on $C_c^\infty(M;\mathbb{R}^k)$. The classical theory
may be quantized in various ways, and we will show that asymptotic measurement
schemes may be obtained for all elements of these various quantizations. 

The \emph{field algebra} may be presented as follows: $\mathcal{F}$ is the complex unital $*$-algebra with generators (`smeared fields') $\varphi(F)$ labelled by $F\in C_c^\infty(M;\mathbb{R}^k)$ and subject to relations
\begin{enumerate}
    \item $F\mapsto \varphi(F)$ is $\mathbb{R}$-linear,
    \item $\varphi(F)=\varphi(F)^*$,
    \item $\varphi(PF) = 0$,
    \item $[\varphi(F),\varphi(G)] = \mathrm{i}E_P(F,G)\openone$,
\end{enumerate}
for arbitrary $F,G\in C_c^\infty(M;\mathbb{R}^k)$. $\mathcal{F}$ is a topological $*$-algebra with respect to the topology $\tau^\varphi$ induced 
by the test function topology such that an arbitrary product of smeared field converges if all the smearing functions do.\footnote{This topology stems from the Borchers-Uhlmann algebra, see for instance \Sect~4.1 in~\cite{Sahlmann2000}. In particular one can see that $\mathcal{F}$ is the quotient of a countable inductive limit of nuclear spaces by a topologically closed two-sided $*$-ideal and hence itself nuclear.} For any region $N\subseteq M$, we define $\mathcal{F}(N)$ to be the subalgebra of $\mathcal{F}$ generated by those $\varphi(F)$ with $F\in C_c^\infty(N;\mathbb{R})$, and endowed with the subspace topology; in particular, $\mathcal{F}(M)=\mathcal{F}$. Although it is common in the literature to allow for complex-valued smearings by $\varphi(F) = \varphi(\Re F)+ \mathrm{i}\varphi(\Im F)$ for $F\in C^\infty_c(M;\mathbb{C}^k)$ it is convenient not to do so here -- though this does not change the algebra, just the way elements are labelled.

We say that $A \in \mathcal{F}$ is localizable (or that $A$ can be localised) in a region $N$, if and only if $A \in \mathcal{F}(N)$; a given element is localizable in many regions, given that
$\mathcal{F}(N_1)\subseteq \mathcal{F}(N_2)$ whenever $N_1\subseteq N_2$ (`isotony') or 
$N_1 \subseteq D(N_2)$ (the `timeslice property').

Finally, we note that $\mathcal{F}$ is a nuclear locally convex topological vector space. In particular, if equation of motion operators $P_1$ and $P_2$ on $C_c^\infty(M;\mathbb{R}^k)$ and $C_c^\infty(M;\mathbb{R}^l)$ respectively induce field algebras $\mathcal{F}_1$ and $\mathcal{F}_2$, the algebraic tensor product $\mathcal{F}_1 \otimes \mathcal{F}_2$ equipped with the unambiguous nuclear topology (see~\cite{Treves:TVS}) is isomorphic to the topological unital $*$-algebra induced by $P_1 \oplus P_2$ on $C_c^\infty(M;\mathbb{R}^{k+l})$ \emph{as a topological $*$-algebra} under the correspondence
\begin{equation}
    \varphi_{P_1\oplus P_2}(f\oplus g) \mapsto 
     \varphi_{P_1}(f)\otimes\openone_{\mathcal{F}_2}  +  \openone_{\mathcal{F}_1} \otimes\varphi_{P_2}(g).
\end{equation}

Generally, a state on a unital $*$-algebra $\mathcal{F}$ is a linear map $\omega:\mathcal{F} \to \mathbb{C}$ that is normalised, $\omega(\openone)=1$, and positive, $\forall A \in \mathcal{F}: \omega(A^* A) \geq 0$.

\subsection{Asymptotic measurement schemes}\label{sec:field_alg_ams}

Suppose, as in \Sect\ref{sec:classical} that the system and probe are described by equation of motions operators $S$ and $P^{ \oplus k}$ respectively, inducing corresponding field algebras $\mathcal{F}_{\mathcal{S}}$ and $\mathcal{F}_\mathcal{P}^{ \otimes k}$, together with the uncoupled 
combination $\mathcal{F}_\mathcal{U}=\mathcal{F}_{\mathcal{S}}\otimes \mathcal{F}_{\mathcal{P}}^{ \otimes k}$
induced by $S\oplus P^{ \oplus k}$ and the coupled theory  $\mathcal{F}_\mathcal{C}$ induced by $T_\lambda$.
We denote the smeared fields generating $\mathcal{F}_\mathcal{S}$ and $\mathcal{F}_\mathcal{P}^{ \otimes k}$
by $\varphi_{\mathcal{S}}(f)$, $\varphi_{\mathcal{P}}^{ \otimes k}(\vec{g})$ respectively\footnote{The superscript $\otimes k$ in $\varphi_{\mathcal{P}}^{ \otimes k}(\vec{g})$ refers to the fact that $\varphi_{\mathcal{P}}^{ \otimes k}$ is the natural field that parametrises the $k$-fold tensor product of field algebras $\mathcal{F}_{\mathcal{P}}^{ \otimes k}$, see Eq.~\eqref{eq_def_varphi_otimes}, and should not be confused with the $k$-fold tensor product of the individual fields parametrising $\mathcal{F}_\mathcal{P}$.},
where
\begin{equation}
    \varphi_{\mathcal{P}}^{ \otimes k}(\vec{g}) =\varphi_\mathcal{P}(g_1) \otimes \openone \otimes \dots \otimes \openone + \openone \otimes \varphi_\mathcal{P}(g_2) \otimes \openone \otimes \dots \otimes \openone + \dots + \openone \otimes \dots \otimes \openone \otimes \varphi_\mathcal{P}(g_k),
    \label{eq_def_varphi_otimes}
\end{equation}
may equivalently be regarded as a multiplet of $k$ scalar fields. The fields of the uncoupled theory are
\begin{equation}
    \varphi_{\mathcal{U}}(f\oplus \vec{g}) = \varphi_{\mathcal{S}}(f)\otimes \openone + \openone\otimes \varphi_{\mathcal{P}}^{ \otimes k}(\vec{g})
\end{equation}
for $f\in C_c^\infty(M;\mathbb{R})$, $\vec{g}\in C_c^\infty(M;\mathbb{R}^k)$. The scattering map associated with the field algebra is obtained from the scattering map of the underlying classical theories, namely that
\begin{equation}
    \Theta_\lambda \varphi_{\mathcal{U}}(F) = \varphi_{\mathcal{U}}(\theta_\lambda F)
\end{equation}
for all $F\in C_c^\infty(M^{ +};\mathbb{R}^{k+1})$. 
 
Following \Sect~5 of~\cite{fewster2018quantum}, the induced observable 
map $\varepsilon_{\sigma}^{\varphi,\lambda R}:\mathcal{F}_{\mathcal{P}}^{ \otimes k}\to \mathcal{F}_{\mathcal{S}}$ obeys
\begin{equation}\label{eq:qinduced}
    \varepsilon_\sigma^{\varphi, \lambda R}(e^{\mathrm{i}x \varphi_\mathcal{P}^{ \otimes k}(\vec{h}_\lambda)} ) = 
    \sigma(e^{\mathrm{i}x \varphi_\mathcal{P}^{ \otimes k}(\vec{g}_\lambda)}) e^{\mathrm{i} x \varphi_\mathcal{S}(\varepsilon^{\mathrm{cl},\lambda R}(\vec{h}_\lambda))}  ,
\end{equation}
as an identity between formal power series in the parameter $x$ with coefficients in $\mathcal{F}_\mathcal{S}$, where $\vec{g}_\lambda$ and 
$f_\lambda=\varepsilon^{\mathrm{cl}, \lambda R}(\vec{h}_\lambda)$ are
as in Eq.~\eqref{eq:convenient} and $\sigma$ is a probe preparation state on $\mathcal{F}_{\mathcal{P}}^{ \otimes k}$. 
Note that we have slightly abused notation by writing $f_\lambda=\varepsilon^{\mathrm{cl}, \lambda R}(\vec{h}_\lambda)$, rather than more properly writing equivalence classes of test functions. This is convenient because the fields are indexed by test functions rather than equivalence classes. 
In~\cite{fewster2018quantum}, the identity~\eqref{eq:qinduced} was used to compute the induced elements obtained from given probe elements. Here, we wish to solve the inverse problem, and begin by rearranging the identity as
\begin{equation}
    \varepsilon_\sigma^{\varphi, \lambda R}(e^{\mathrm{i}x \varphi_\mathcal{P}^{ \otimes k}(\vec{h}_\lambda)} / \sigma(e^{\mathrm{i}x \varphi_\mathcal{P}^{ \otimes k}(\vec{g}_\lambda)})) = e^{\mathrm{i} x \varphi_\mathcal{S}(\varepsilon^{\mathrm{cl},\lambda R}(\vec{h}_\lambda))},
    \label{eq_fields_induced_obs_formal_series}
\end{equation}
again as an identity of formal power series, making use of the fact that the formal power series $\sigma(e^{\mathrm{i}x \varphi_\mathcal{P}^{ \otimes k}(\vec{g}_\lambda)})$ has an inverse because its constant term is nonzero, and also using the linearity of $\varepsilon_\sigma^{\varphi, \lambda R}$. By equating powers of $x$ we see
that every power of $\varphi_\mathcal{S}(\varepsilon^{\mathrm{cl},\lambda R}(\vec{h}_\lambda))$ admits a measurement scheme, because
\begin{equation}\label{eq:powms}
    \varphi_\mathcal{S}(\varepsilon^{\mathrm{cl},\lambda R}(\vec{h}_\lambda))^n = 
     \varepsilon_\sigma^{\varphi, \lambda R}\left(
     (-\mathrm{i})^{n} \dv[n]{}{x} \frac{e^{\mathrm{i}x\varphi_\mathcal{P}^{\otimes k}(\vec{h}_\lambda)}}{\sigma(e^{\mathrm{i} x \varphi_\mathcal{P}^{\otimes k}(\vec{g}_\lambda)})}\Bigg|_{x=0}
     \right),
\end{equation}
interpreting the differentiation and evaluation at zero appropriately to formal power series. Furthermore, the classical ${ \tau^{\mathrm{cl}}}$-asymptotic measurement scheme $H_\lambda^{\mathrm{cl}}$ for $[f]_S$ in Eq.~\eqref{eq_classical_ams} immediately induces a $\tau^\varphi$-asymptotic measurement scheme for each power
$\varphi_\mathcal{S}(f)^n$, owing to the
$\tau^\varphi$-continuity of any power of smeared fields.
For example,
\begin{equation}\label{eq:field_scheme}
   \varphi_\mathcal{S}(\varepsilon^{\mathrm{cl},\lambda R} (\vec{h}_\lambda))= \varepsilon^{\varphi, \lambda R}_{\sigma}(\varphi_\mathcal{P}^{ \otimes k}(\vec{h}_\lambda) - \sigma(\varphi_\mathcal{P}^{ \otimes k}(\vec{g}_\lambda))\openone) 
\end{equation}
where $\vec{g}_\lambda$ is as in Eq.~\eqref{eq:convenient}, which shows that $(\mathcal{F}_\mathcal{P}^{ \otimes k}, \varepsilon_{\sigma}^{\varphi,\lambda R},\varphi_\mathcal{P}^{ \otimes k}(\vec{h}_\lambda) - \sigma(\varphi_\mathcal{P}^{ \otimes k}(\vec{g}_\lambda))\openone)$ provides a measurement scheme for $\varphi_\mathcal{S}(\varepsilon^{\mathrm{cl},\lambda R} (\vec{h}_\lambda))$,
and thus providing a $\tau^\varphi$-asymptotic measurement scheme for $\varphi_\mathcal{S}(f)$ as
$\lambda\to 0$.

In fact, we can go further. The general element of $\mathcal{F}_{\mathcal{S}}$ is a complex linear combination of products of generators. Using the commutation relations, every element of $\mathcal{F}_{\mathcal{S}}$ can be expressed
as a finite linear combination of symmetrised products of generators. 
Next, by the multi-linear generalization of the polarization identity, see e.g.~Eq.~(A.4) in~\cite{Thomas2014}, every symmetrised $n$-fold product of $\varphi_\mathcal{S}(f_1),...,\varphi_\mathcal{S}(f_n)$ can be written as
\begin{equation}
\begin{aligned}
    &\sum\limits_{\pi \in S_n} \varphi_\mathcal{S}(f_{\pi(1)}) ... \varphi_\mathcal{S}(f_{\pi(n)})\\
    &= \frac{1}{2^n} \sum\limits_{\epsilon_1,...,\epsilon_n=0}^1 (-1)^{\epsilon_1+...+\epsilon_n} \qty((-1)^{\epsilon_1} \varphi_\mathcal{S}(f_1)+ ...+ (-1)^{\epsilon_n} \varphi_\mathcal{S}(f_n))^n\\
    &=\frac{1}{2^n} \sum\limits_{\epsilon_1,...,\epsilon_n=0}^1 (-1)^{\epsilon_1+...+\epsilon_n} \qty(\varphi_\mathcal{S}((-1)^{\epsilon_1} f_{1}+ ...+ (-1)^{\epsilon_n} f_n))^n.
    \label{eq_multilin_pol_id}
\end{aligned}
\end{equation}
Therefore every element of $\mathcal{F}_\mathcal{S}$ may be written as a complex linear combination of powers of generators
\begin{equation}\label{eq:typical_field_obs}
    A = \sum\limits_{j=1}^k c_j \varphi_\mathcal{S}(f_j)^{n_j}
\end{equation} 
with real-valued $f_j$ and $c_j \in \mathbb{C}$. We will now construct an asymptotic measurement scheme for $A$ using $k$ probe fields. For every $j$ let $H_\lambda^{\mathrm{cl},j} = \qty(\mathcal{C}_\mathcal{P}, \varepsilon^{\mathrm{cl}, \lambda R_j}, [h^j_\lambda]_P)$ be the classical asymptotic measurement scheme for $\qty[f_j]_S$ as in Theorem~\ref{theo_comb_classical_mscheme}. Also as in that result, we consider a probe consisting of $k$ fields with coupling $R=(R_1,\ldots,R_k)^T$. Let us hence introduce the notation
\begin{equation}
    \vec{h}^j:= 0 \oplus \dots 0\oplus h^j \oplus 0 \dots \oplus 0 \equiv h^j \vec{e}^j \in C_c^\infty(M;\mathbb{R}^k),
\end{equation}
where $\vec{e}^j$ is the $j^\mathrm{th}$ standard basis vector in $\mathbb{R}^k$; 
we also write $\vec{h}_\lambda^j:= \vec{h}^j/\lambda$. Define $f^j_\lambda$ and
  $\vec{g}^j_\lambda$ by 
\begin{equation}
    \mqty(f^j_\lambda \\ \vec{g}_\lambda^j):= \mqty(0 \\ \vec{h}^j_\lambda) - (T_\lambda - S\oplus P^{\oplus l}) E_{T_\lambda}^-\mqty(0 \\ \vec{h}^j_\lambda),
\end{equation}
(cf.\ Eq.~\eqref{eq:convenient}) so that $f^j_\lambda = \varepsilon^{\mathrm{cl}, \lambda R}(\vec{h}^j_\lambda)$ for each $j$. It now follows from~\eqref{eq:powms}, applied to each $\vec{h}^j_\lambda$ in turn, and the linearity of $\varepsilon_{\sigma}^{\varphi, \lambda R}$, that
\begin{equation}
    H_\lambda^\varphi:=\qty(\mathcal{F}_\mathcal{P}^{\otimes k},\varepsilon_{\sigma}^{\varphi, \lambda R}, \sum\limits_{j=1}^k c_j (-\mathrm{i})^{n_j} \dv[n_j]{}{x} \frac{e^{\mathrm{i}x\varphi_\mathcal{P}^{\otimes k}(\vec{h}^j_\lambda)}}{\sigma(e^{\mathrm{i} x \varphi_\mathcal{P}^{\otimes k}(\vec{g}^j_\lambda)})}\Bigg|_{x=0} ),
    \label{eq_field_mscheme}
\end{equation}
is a measurement scheme for 
\begin{equation}
\sum\limits_{j=1}^k c_j
\varphi_\mathcal{S}(\varepsilon^{\mathrm{cl},\lambda R}(\vec{h}_\lambda^j))^n =
\sum\limits_{j=1}^k c_j \varphi_\mathcal{S}(f^j_\lambda)^{n_j}
\label{eq_field_mscheme_induced}
\end{equation} 
and hence $(H_\lambda^\varphi)$ is a $\tau^\varphi$-asymptotic measurement scheme
for $A$ in the limit $\lambda\to 0$, with coupling in $N$ and processing region $L$.

Finally, we note that any Hermitian $A$ can be written as a finite \emph{real} linear combination of symmetrised products of generators, and, consulting Eq.~\eqref{eq_multilin_pol_id}, also in the form of Eq.~\eqref{eq:typical_field_obs} with \emph{real} $c_j$'s. It is then immediate that  $(H_\lambda^\varphi)$ is a \emph{Hermitian} $\tau^\varphi$-asymptotic measurement scheme
for $A$ in the limit $\lambda\to 0$, with coupling in $N$ and processing region $L$.

In summary we have proved the following theorem.

\begin{theo}\label{theo_existence_of_asymp_measurement_schemes_field_algebra}
Let $\mathcal{F}_\mathcal{S}$ be equipped with $\tau^{\varphi}$. Then, for every precompact region $N$, for every $A \in \mathcal{F}_\mathcal{S}(N)$ and for every region $L \subseteq M \setminus J^-(\overline{N})$ such that $N \subseteq D(L)$ there exists a $\tau^\varphi$-asymptotic measurement scheme for $A$ with coupling in $N$ and processing region $L$. If $A$ is Hermitian, then the $\tau^\varphi$-asymptotic measurement scheme can be chosen to be Hermitian as well.
\end{theo}

According to Theorem~\ref{theo_faster_approx} for every classical system observable and every $k$ there exists a classical asymptotic measurement scheme of order $2k$. As we show in  Appendix~\ref{sec_appendix_faster_approx}, the argument generalises and yields the following theorem.

\begin{theo}
For every precompact region $N$, every admissible processing region $L$, every $k$ and every $\varphi_\mathcal{S}(f) \in \mathcal{F}_\mathcal{S}(N)$ there is an asymptotic measurement scheme for $\varphi_\mathcal{S}(f)$ of order $2k$.
\label{theo_faster_approx_fields}
\end{theo}

Recall from \Sect\ref{sec:effort}, that the notion of \emph{order} of an asymptotic measurement scheme utilized a seminorm $\mathrm{eff}$ being a measure of \emph{effort}. In the context of Theorem~\ref{theo_faster_approx_fields}, a candidate for such a seminorm is given as follows. For a probe preparation state $\sigma$ consider the following positive semi-definite\footnote{This follows from the Cauchy-Schwarz inequality and the fact that $\sigma(A^*)=\overline{\sigma(A)}$.} Hermitian sesquilinear form
\begin{equation}
    \mathrm{cov}_\sigma(A,B):= \sigma(A^* B) - \sigma(A^*) \sigma(B),
\end{equation}
which induces the seminorm
\begin{equation}
    \mathrm{eff}_\sigma(A):= \sqrt{\mathrm{cov}_\sigma(A,A)}.
    \label{eq_eff_sigma}
\end{equation}
See \Sect\ref{sec_variance} for a detailed account on the interpretation of $\mathrm{eff}_\sigma(A)$.\\

Before we move on to a discussion of asymptotic measurement schemes for the Weyl algebra, let us note that each of $\mathcal{F}_\mathcal{S}, \mathcal{F}_\mathcal{P}^{\otimes k}, \mathcal{F}_\mathcal{U}$ and $\mathcal{F}_\mathcal{C}$ admits a global $\mathbb{Z}_2$ symmetry defined by its action on the generators as $ \varphi_*(F) \mapsto -\varphi_*(F)$ where $* \in \qty{\mathcal{S}, \mathcal{U}, \mathcal{P}, \mathcal{C}}$ and $F$ chosen appropriately. This symmetry may be regarded as a global \emph{gauge} symmetry, upon which only the Hermitian elements that are invariant under this transformation are deemed \emph{observables}. It follows by direct inspection of Eq.~\eqref{eq:qinduced} that $\varepsilon_\sigma^{\varphi, \lambda R}(B)$ is a $\mathbb{Z}_2$-gauge-invariant element of $\mathcal{F}_\mathcal{S}$ if $B$ is a $\mathbb{Z}_2$-gauge-invariant element of $\mathcal{F}_\mathcal{P}^{\otimes k}$ and $\sigma$ is a gauge-invariant state. (This also follows from general results proved in~\cite{fewster2018quantum}.) More importantly, if $\sigma$ is gauge invariant (for instance quasi-free with vanishing one-point function) and $A$ from above is gauge-invariant (i.e., $n_j$ is even for every $j$), then also the elements in Eq.~\eqref{eq_field_mscheme_induced} are gauge invariant and moreover the probe observables for the measurement schemes $H_\lambda^\varphi$ in Eq.~\eqref{eq_field_mscheme} are gauge invariant. In summary, every gauge-invariant $A \in \mathcal{F}_\mathcal{S}$ admits a gauge-invariant asymptotic measurement scheme, i.e., an asymptotic measurement scheme with gauge-invariant probe element.

\section{Quantum asymptotic measurement schemes: the Weyl algebra}\label{sec_qams_Weyl}

\subsection{The Weyl algebra}\label{sec:Weyl_algebra}

As in \Sect\ref{sec:field_algebra}, consider a classical equation of motion operator $P$ on $C_c^\infty(M;\mathbb{R}^k)$. The classical theory $\mathcal{C}$ induced by $P$ also 
admits a CCR-$C^*$-quantization in terms of abstract Weyl generators $W(f)$ indexed by $f \in C_c^\infty(M;\mathbb{R}^k)$. They fulfill
\begin{enumerate}
    \item $W(f)^* = W(-f)$,
    \item $W(Pf)=\openone$,
    \item $W(f) W(g) = e^{-\frac{\mathrm{i}}{2} E_P(f,g)} W(f+g)$,
\end{enumerate}
so in particular $W(0) = \openone$ and $W(f)^{*}= W(f)^{-1}$; moreover, $W(f)=W(g)$ whenever $f$ and $g$ are equivalent. The unital $*$-algebra spanned by all finite $\mathbb{C}$-linear combinations, products and adjoints of Weyl generators, subject to the above relations, can be equipped with a unique $C^*$-norm. Its completion in this norm defines the (global) CCR-$C^*$-algebra for fixed $M$, which we will denote simply by $\mathcal{A}$. 

Furthermore, to any region $N \subseteq M$, we can associate the unital $C^*$-subalgebra $\mathcal{A}(N)$ of $\mathcal{A}$ generated by $W(f)$'s indexed by $f$ with support in the region $N$, in particular $\mathcal{A}= \mathcal{A}(M)$. Again, $A\in \mathcal{A}$ is said to be localizable in a region $N$, if and only if $A \in \mathcal{A}(N)$, and just as with the field algebra, any given element is localizable in many regions.

Finally, we note that $\mathcal{A}$ is a nuclear $C^*$-algebra. In particular, if equation of motion operators $P_1$ and $P_2$ on $C_c^\infty(M;\mathbb{R}^k)$ and $C_c^\infty(M;\mathbb{R}^l)$ respectively induce CCR-$C^*$-algebras $\mathcal{A}_1$ and $\mathcal{A}_2$, the unique $C^*$ tensor product $\mathcal{A}_1 \overline{\otimes} \mathcal{A}_2$ is isomorphic to the CCR-$C^*$-algebra induced by $P_1 \oplus P_2$ on $C_c^\infty(M;\mathbb{R}^{k+l})$ \emph{as a $C^*$-algebra}.\footnote{See for instance Proposition 18.1-18 in~\cite{Honegger2015}, and Theorem 10.10 in~\cite{Evans1977} for the nuclearity of the CCR-$C^*$-algebras.}

The advantage of using the $C^*$-quantization (instead of for instance the $*$-quantization described in~\cite{fewster2018quantum}) is the ability to utilize the well-developed $C^*$-representation theory and to consider the associated von Neumann algebras as we will do in Corollary~\ref{corol_vN}. Formally, the $W(f)$'s can be viewed as ``exponentiated smeared quantum fields'' 
\begin{equation}
W(f) = e^{\mathrm{i} (\varphi_1(f^1) + ... + \varphi_k(f^k))},
\end{equation}
for $f=(f^1,...,f^k)^T$. This can be made rigorous, for instance, in the GNS representation of \emph{analytic} states.

For the convenience of the reader, an introduction to analytic states, quasi-free states, field operators and the GNS representation can be found in Appendix~\ref{appendix_conv_field_ops}. In brief, states on $\mathcal{A}$ may be used to connect the above $C^*$-quantization to field operators and also to von Neumann algebras in the following way.

\begin{itemize}
    \item In a given GNS representation $\pi$ of an analytic state on the algebra $\mathcal{A}$ there is a densely defined self-adjoint field operator $\varphi^\pi(f)$ for every $f \in C_c^\infty(M;\mathbb{R}^k)$ such that
\begin{equation}
    \pi(W(f)) = e^{\mathrm{i}\varphi^\pi(f)}.
\end{equation}
Note that $\varphi^\pi$ can also be regarded as a multiplet of $k$ scalar fields.
\item For a given fixed state $\omega$ one can define an AQFT of von Neumann algebras by associating to every region $N$ the weak closure of $\pi_\omega[\mathcal{A}(N)]$ in $BL(\mathcal{H}_\omega)$, or equivalently (by von Neumann's bicommutant Theorem)
\begin{equation}
    \overline{\pi_\omega[\mathcal{A}(N)]}^w=\qty(\pi_\omega[\mathcal{A}(N)])'' \subseteq BL(\mathcal{H}_\omega). 
\end{equation}
\end{itemize}

It is well known that the $C^*$-norm topology $\tau_{\|\cdot\|}$ is too strong for physical purposes -- for example, all differences of distinct Weyl generators have norm $2$, see for instance Proposition 7~in~\cite{Baer2009} -- hence we will define a more useful topology $\tau$ by reference to the strong$^*$ operator topologies in suitable GNS representations. The resulting topology will be used in our discussion of asymptotic measurement schemes.\\

To prepare for the definition of $\tau$, we remind the reader that a topology $\tau_A$ on a set $X$ is called weaker (or coarser, or smaller) than topology $\tau_B$ on $X$, if and only if $\tau_A \subseteq \tau_B$. In this case one also says that $\tau_B$ is stronger (or finer, or larger) than $\tau_A$. The weakest topology is $\qty{\emptyset, X}$, the strongest topology is the power set of $X$. In particular, every set that is $\tau_A$-open is also $\tau_B$ open, but $\tau_B$ has (possibly) more open sets, so it is more difficult for a net (Moore-Smith sequence) to converge. Every net $(a_\alpha)_\alpha$ that converges to a point $a$ in the stronger topology $\tau_B$ also converges in the weaker topology $\tau_A$, but the converse does not hold in general. If $Z \subseteq X$ is $\tau_B$-dense, then it is also $\tau_A$-dense. Any map $f:X \to Y$ for a topological space $Y$ that is continuous with respect to the topology $\tau_A$ is also continuous with respect to the topology $\tau_B$, but the converse does not hold in general. In summary, stronger topologies have more open sets, fewer convergent nets, fewer dense subsets, and more continuous functions into other topological spaces.

Although the norm topology $\tau_{\|\cdot\|}$ on $\mathcal{A}_\mathcal{S}$ is too strong, every state $\omega$ on $\mathcal{A}_\mathcal{S}$ induces three further interesting topologies via its GNS representation (see Appendix~\ref{appendix_conv_field_ops}), which are both weaker than $\tau_{\|\cdot\|}$.

\begin{mydef}
Let $\omega$ be a state on $\mathcal{A}_\mathcal{S}$ with GNS representation $\pi_\omega: \mathcal{A}_\mathcal{S}\to BL(\mathcal{H}_\omega)$. Then we define
\begin{enumerate}
    \item the $\pi_\omega$-weak operator topology $\tau_w^\omega$ on $\mathcal{A}_\mathcal{S}$ as the weakest topology such that $\pi_\omega: \mathcal{A}_\mathcal{S} \to \qty(BL(\mathcal{H}_\omega), \tau_w)$ is continuous, where $\tau_w$ is the weak operator topology,
    \item the $\pi_\omega$-strong operator topology $\tau_{st}^\omega$ on $\mathcal{A}_\mathcal{S}$ as the weakest topology such that $\pi_\omega: \mathcal{A}_\mathcal{S} \to \qty(BL(\mathcal{H}_\omega), \tau_{st})$ is continuous, where $\tau_{st}$ is the strong operator topology, and
    \item the $\pi_\omega$-strong$^*$ operator topology $\tau_{st^*}^\omega$ on $\mathcal{A}_\mathcal{S}$ as the weakest topology such that $\pi_\omega: \mathcal{A}_\mathcal{S} \to \qty(BL(\mathcal{H}_\omega), \tau_{st^*})$ is continuous, where $\tau_{st^*}$ is the strong$^*$ operator topology.
\end{enumerate}
Equipped with either $\tau_w^\omega$ or $\tau_{st}^\omega$ or $\tau_{st^*}^\omega$, $\mathcal{A}_\mathcal{S}$ is a locally convex topological vector space and 
\begin{enumerate}
    \item $\tau_w^\omega$ is generated by the family of seminorms $|\braket{x|\pi_\omega (\cdot)y}_\omega|$ for $x,y \in \mathcal{H}_\omega$,
    \item $\tau_{st}^\omega$ is generated by the family of seminorms $\|\pi_\omega (\cdot)x\|_\omega$ for $x \in \mathcal{H}_\omega$, and
     \item $\tau_{st^*}^\omega$ is generated by the family of seminorms $\|\pi_\omega (\cdot)x\|_\omega +\|\pi_\omega (\cdot)^*x\|_\omega$ for $x \in \mathcal{H}_\omega$.
\end{enumerate}
It holds that
\begin{equation}
    \tau_w^\omega \subseteq \tau_{st}^\omega \subseteq \tau_{st^*}^\omega \subseteq \tau_{\|\cdot\|}.
\end{equation}
\label{def_operator_topologies}
\end{mydef}

We now introduce the topology $\tau$.
\begin{mydef}
We define $\mathfrak{S}_c$ to be the set of all states on $\mathcal{A}_\mathcal{S}$ such that $\omega \circ W_\mathcal{S}: C_c^\infty(M;\mathbb{R}) \to \mathbb{C}$ is continuous with respect to the 
standard topology on $C_c^\infty(M;\mathbb{R})$. Then $\tau$ is defined to be the locally convex topology generated by the seminorms $\|\pi_\omega (\cdot)x\|_\omega {+\|\pi_\omega (\cdot)^*x\|_\omega}$ for $x \in \mathcal{H}_\omega$ and $\omega \in \mathfrak{S}_c$, i.e., the weakest topology that contains every $\tau_{st^*}^\omega$ for $\omega \in \mathfrak{S}_c$.
\label{def_tau}
\end{mydef}

\paragraph*{Remark:} $\mathfrak{S}_c$ contains in particular all quasi-free states with distributional two-point function, see Appendix~\ref{appendix_conv_field_ops}. Moreover, every state in $\mathfrak{S}_c$ is \emph{regular}, see for instance \Sect~5.2.3 in~\cite{BratteliRobinsonII1997}.\\

The usefulness of this topology is illustrated by the following lemma.

\begin{lem}
For every $\omega \in \mathfrak{S}_c$
\begin{equation}
    \tau_w^\omega \subseteq \tau_{st}^\omega {\subseteq \tau_{st^*}^\omega} \subseteq \tau \subseteq \tau_{\|\cdot\|},
\end{equation}
and $\omega: (\mathcal{A}_\mathcal{S},\tau) \to \mathbb{C}$ as well as the GNS representation $\pi_\omega: (\mathcal{A}_\mathcal{S},\tau) \to BL(\mathcal{H}_\omega)$ is continuous, where $BL(\mathcal{H}_\omega)$ is equipped either with its strong$^*$ or strong or weak operator topology.

Furthermore, $\tau$ is $*$-compatible and the map 
\begin{equation}
    \begin{aligned}
    W_\mathcal{S}: C_c^\infty(M;\mathbb{R}) &\to \qty(\mathcal{A}_\mathcal{S}, \tau)\\
    f &\mapsto W_\mathcal{S}(f)
    \end{aligned}
\end{equation}
is continuous and for every $\tau$-dense $\mathcal{I} \subseteq \mathcal{A}_\mathcal{S}$,
\begin{equation}
    \overline{\pi_\omega\qty[\mathcal{I}]}^w = \overline{\pi_\omega\qty[\mathcal{A}_\mathcal{S}]}^w = \qty(\pi_\omega\qty[\mathcal{A}_\mathcal{S}])'',
    \label{eq_dense_vN}
\end{equation}
where $\overline{\cdot}^w$ denotes the closure in $\qty(BL(\mathcal{H}_\omega), \tau_w)$.
\label{lem_W_tau_cont}
\end{lem}

\begin{proof}
The first sentence follows straightforwardly from Definition~\ref{def_operator_topologies} and Definition~\ref{def_tau}.

Recall that a function $F$ from some topological space $X$ into $\qty(\mathcal{A}_\mathcal{S},\tau)$ is continuous if and only if $F$ is continuous into  $\qty(\mathcal{A}_\mathcal{S},\tau_{st^*}^\omega)$ for every $\omega \in \mathfrak{S}_c$, which holds if and only if $\pi_\omega \circ  F : X \to \qty(BL(\mathcal{H}_\omega), \tau_{st^*})$ is continuous for every $\omega \in \mathfrak{S}_c$. It then follows by Lemma~\ref{lem_strong_convergence_id} that $W_\mathcal{S}$ is continuous. The $\tau$-continuity of the $*$-operation is easy to see from the definition of $\tau_{st^*}^\omega$.

Finally, since $\tau$ is stronger than $\tau_w^\omega$, $\mathcal{I}$ is $\tau_w^\omega$-dense in $\mathcal{A}_\mathcal{S}$ so $\pi_\omega \qty[\mathcal{I}]$ is $\tau_w$-dense in $\pi_\omega\qty[\mathcal{A}_{ \mathcal{S} }]$. Hence $\pi_\omega \qty[\mathcal{I}]$ is also $\tau_w$-dense in $\overline{\pi_\omega\qty[\mathcal{A}_\mathcal{S}]}^w = \qty(\pi_\omega\qty[\mathcal{A}_\mathcal{S}])''$.
\end{proof}

\subsection{Asymptotic measurement schemes}

We return to the situation described in \Sect\ref{sec:classical_combinations}, of a
system consisting of a single linear real system field $\varphi_\mathcal{S}$ with equation of motion operator $S$ on $C^\infty(M;\mathbb{R})$ and $k$ linear real probe fields $\qty(\varphi_\mathcal{P})_j$ with combined equation of motion operator $P^{ \oplus k}$ on $C^\infty(M;\mathbb{R}^{k})$, with $S$ and $P^{ \oplus k}$ both linear, second order, normally hyperbolic and formally self-adjoint. Then $S$ and $P^{ \oplus k}$ determine system and probe theories $\mathcal{A}_\mathcal{S}$ and $\mathcal{A}_{\mathcal{P}}^{\otimes k}$ with Weyl generators denoted by $W_\mathcal{S}$ and $W_{\mathcal{P}}^{\otimes k}$ respectively, { where $\mathcal{A}_\mathcal{P}$ is the theory induced by $P$. As mentioned before,} $\mathcal{A}_\mathcal{S} \overline{\otimes} \mathcal{A}_{\mathcal{P}}^{\otimes k}$, is isomorphic to the AQFT $\mathcal{A}_{\mathcal{U}}$ obtained from $S\oplus P^{ \oplus k}$
under the isomorphism sending
\begin{equation}
    W_\mathcal{S}(f) \otimes W_\mathcal{P}^{ \otimes k}(\vec{h}) \mapsto W_\mathcal{U}(f \oplus \vec{h}),
\end{equation}
for $f \in C_c^\infty(M;\mathbb{R})$, $\vec{h} \in C_c^\infty(M;\mathbb{R}^{k})$, where the Weyl generators of $\mathcal{A}_\mathcal{U}$ are denoted $W_\mathcal{U}$. This theory describes the \emph{free, uncoupled} combination of system and probe. Meanwhile, the AQFT $\mathcal{A_C}$ induced by 
$T{_\lambda}$ will be the coupled variant of the free combination of the system and the probe theory according to the FV framework.

The associated scattering map $\Theta_\lambda$ was derived in~\cite{fewster2018quantum} and acts on Weyl generators by 
\begin{equation}\label{eq:scattering}
    \Theta{_\lambda} W_\mathcal{U}(F) = W_\mathcal{U}(\theta_\lambda F),
\end{equation}
with $\theta_\lambda$ as in~\eqref{eq:unvartheta_def} { for $F \in C_c^\infty(M^+; \mathbb{R}^{k+1})$}. The induced observables may now be described.
We consider the specific case where the probe element is comprised of one Weyl generator for each of the $k$ probe fields, tensored together. This corresponds to a Weyl generator $W_\mathcal{P}^{ \otimes k}(\vec{h}_\lambda)$ of $\mathcal{A}_\mathcal{P}^{ \otimes k}$ where $\vec{h}_\lambda \in C_c^\infty(M^+; \mathbb{R}^{k})$ is supported in the `out' region. 
Using~\eqref{eq:scattering} and~\eqref{eq:convenient} (and its notation) we have
\begin{equation}\label{eq:prevarepsGH}
    \Theta_\lambda \qty(\openone \otimes W_\mathcal{P}^{ \otimes k}(\vec{h}_\lambda)) = W_\mathcal{S}(f_\lambda)\otimes W_\mathcal{P}^{ \otimes k}(\vec{g}_\lambda) .
\end{equation}
Consequently, the corresponding induced element is
\begin{equation}\label{eq:varepsGH}
    \varepsilon_{\sigma}^{W, \lambda R}(W_\mathcal{P}^{ \otimes k}(\vec{h}_\lambda)) = \sigma(W_\mathcal{P}^{ \otimes k}(\vec{g}_\lambda)) W_\mathcal{S}(f_\lambda),
\end{equation}
which is a clear analogue of Eq.~\eqref{eq_fields_induced_obs_formal_series}. In particular, $\varepsilon_{\sigma}^{W, \lambda R}(W_\mathcal{P}^{ \otimes k}(\vec{h}_\lambda))$ can be localised in any region containing the support of $R$.

The main result of this section is that every element of the Weyl algebra $\mathcal{A}_\mathcal{S}$ admits an asymptotic measurement scheme with respect to the topology $\tau$ from Definition~\ref{def_tau}.
 
\begin{theo} 
Let $\tau$ be the topology on $\mathcal{A}_{\mathcal{S}}$ given in Definition~\ref{def_tau}. Then, for every precompact region $N$, every $A \in \mathcal{A}_\mathcal{S}(N)$ and every region $L \subseteq M \setminus J^-(\overline{N})$ such that $N \subseteq D(L)$ there exists a $\tau$-\emph{asymptotic measurement scheme} for $A$ with coupling in $N$ and processing region $L$. If $A$ is Hermitian, then the $\tau$-asymptotic measurement scheme can be chosen to be Hermitian as well.
\label{theo_asympt_m_scheme}
\end{theo}

To prove Theorem~\ref{theo_asympt_m_scheme}, we will show that we can find an asymptotic measurement scheme for every element in a $\tau$-dense subset of $\mathcal{A}_\mathcal{S}$ and hence, by Lemma~\ref{lem_closed}, for
all elements of $\mathcal{A}_\mathcal{S}$. It follows immediately from Theorem~\ref{theo_asympt_m_scheme} and Lemma~\ref{lem_closed} that the set of elements $\mathcal{I}$ of $\mathcal{A}_\mathcal{S}$ admitting a \emph{bona fide} measurement scheme is $\tau$-dense in $\mathcal{A}_\mathcal{S}$. Moreover, we have the following corollary to Theorem~\ref{theo_asympt_m_scheme}.

\begin{corol}
For every $\omega \in \mathfrak{S}_c$ and every element $A$ in the von Neumann algebra $\qty(\pi_\omega\qty[\mathcal{A}_\mathcal{S}])''$ there exists a net $(A_\alpha)_{\alpha} \subseteq \mathcal{A}_\mathcal{S}$ and measurement schemes $H_\alpha$ for each $A_\alpha$ such that $\pi_ \omega(A_\alpha) \to A$ in the weak topology. If $A$ is Hermitian, then every $H_\alpha$ can be chosen to be Hermitian as well.
\label{corol_vN}
\end{corol}

\begin{proof}
By Eq.~\eqref{eq_dense_vN} in Lemma~\ref{lem_W_tau_cont}, $\pi_\omega[\mathcal{I}]$ is weakly dense in $\qty(\pi_\omega\qty[\mathcal{A}_\mathcal{S}])''$. Due to the weak continuity of the star operation, the Hermitian elements in $\pi_\omega[\mathcal{I}]$ are weakly dense in the Hermitian elements of $\qty(\pi_\omega\qty[\mathcal{A}_\mathcal{S}])''$.
\end{proof}

In this sense one could call $(H_\alpha)_\alpha$ a $\tau_w$-asymptotic measurement scheme for $A$, whether or not the measurement scheme can be fully implemented at the level of von Neumann algebras.\footnote{In fact, by the causal propagation property of the Hadamard form we see that $\Theta$ preserves the class of quasi-free Hadamard states on $\mathcal{A}_\mathcal{S} \otimes \mathcal{A}_\mathcal{P}$. It then follows from Theorem 3.6 in~\cite{Verch1997} and Theorem 2.4.26 (1) in~\cite{BratteliRobinsonI1987} that $\Theta$, and hence the measurement schemes $H_\alpha$, are implementable on the von Neumann algebra of any precompact region $\tilde{M} \subseteq M$ that contains $N$ and the processing region $L$.}\\

We will now proceed to prove Theorem~\ref{theo_asympt_m_scheme}. The first observation is that -- just as for the field algebra -- the classical asymptotic measurement schemes provided by Theorem~\ref{theo_classical_asympt_m_scheme} lift immediately to give asymptotic measurement schemes for any multiple of a Weyl generator in the quantum theory. To see this, let $W_\mathcal{S}(f) \in \mathcal{A}_\mathcal{S}(N)$ for a precompact region $N$ and let $L$ be a region in $ M \setminus J^-(\overline{N})$ such that $N \subseteq D(L)$. Let $H_\lambda^\mathrm{cl}$ be the classical asymptotic measurement scheme for $\qty[f]_S$ from above and let $g_\lambda$ be as in Eq.~\eqref{eq:convenient}. We choose a fixed probe state $\sigma$ so that for every $\lambda:$ $\sigma(W_\mathcal{P}(g_\lambda)) \neq 0$ (for instance by choosing a quasi-free state, see Appendix~\ref{appendix_conv_field_ops}). Then, for every $c \in \mathbb{C}$ the collection of quantum measurement schemes
\begin{equation}
    \qty(\mathcal{A}_\mathcal{P},\varepsilon_{\sigma}^{W, R_\lambda}, c \frac{W_\mathcal{P}(h_\lambda)}{\sigma(W_\mathcal{P}(g_\lambda))} )
    \label{eq_single_mscheme_quant}
\end{equation}
is a $\tau$-asymptotic measurement scheme for $c W_\mathcal{S}(f)$ with coupling in $N$ and processing region $L$. This follows immediately, by observing that
\begin{equation}
     \varepsilon_{\sigma}^{W, R_\lambda}\qty(c \frac{W_\mathcal{P}(h_\lambda)}{\sigma(W_\mathcal{P}(g_\lambda))} ) = c W_\mathcal{S}(\varepsilon^{\mathrm{cl}, \lambda R}(h_\lambda)) \longrightarrow c W_\mathcal{S}(f),
     \label{eq_relationship_vareps_cl_quant}
\end{equation}
according to Eq.~\eqref{eq:varepsGH}, Eq.~\eqref{eq_classical_varepsilon}, Lemma~\ref{lem_W_tau_cont} and Lemma~\ref{lem_approx_single_test_multiple_probe_classical}. Note that we have again slightly abused notation by writing $\varepsilon^{\mathrm{cl}, \lambda R}(h_\lambda)=f_\lambda$.\\
 
In order to construct an asymptotic measurement of a general element in
\begin{equation}\Weyl_k(N) := \left\{\sum_{j=1}^k \alpha_j W_\mathcal{S}(f_j)| \alpha_j \in \mathbb{C}, f_j \in C_c^\infty(N;\mathbb{R)}\right\},
\end{equation} 
we again utilize the combination of $k$ asymptotic measurement schemes of the kind constructed in Theorem~\ref{theo_comb_classical_mscheme}. The argument above, of course, corresponds to the case $k=1$. Consider $A=\sum\limits_{j=1}^k c_j W_\mathcal{S}(f_j)\in \Weyl_k(N)$ and construct the coupling $R$ for $k$ coupled fields exactly as described after~\eqref{eq:typical_field_obs}.
Fix any probe state $\sigma$ such that $\sigma(W_\mathcal{P}^{ \otimes k}(\vec{g}^j_\lambda)) \neq 0$ for all $j$ (for instance a quasi-free probe state on $\mathcal{A}_\mathcal{P}^{\otimes k}$ -- see Appendix~\ref{appendix_conv_field_ops}). Owing to the formal analogy between~\eqref{eq:varepsGH} and~\eqref{eq:qinduced}, we may read off immediately that  
\begin{equation}
\begin{aligned}
    &\lim_{\lambda \to 0} \varepsilon_{\sigma}^{W, \lambda R}\qty(\sum\limits_{j=1}^k c_j \frac{W_\mathcal{P}^{\otimes k}(\vec{h}^j_\lambda)}{\sigma(W_\mathcal{P}^{\otimes k}(\vec{g}^j_\lambda))} ) = \sum\limits_{j=1}^k c_j W_\mathcal{S}(f_j),
\end{aligned}
\end{equation}
where we have used the $\tau$-continuity of Weyl generators. 
Thus,
\begin{equation}
    H_\lambda:=\qty(\mathcal{A}_\mathcal{P}^{\otimes k},\varepsilon_{\sigma}^{W, \lambda R}, \sum\limits_{j=1}^k c_j \frac{W_\mathcal{P}^{\otimes k}(\vec{h}^j_\lambda)}{\sigma(W_\mathcal{P}^{\otimes k}(\vec{g}^j_\lambda))} ),
\end{equation}
is a $\tau$-asymptotic measurement scheme for $A$, in the limit $\lambda\to 0$, with coupling in $N$ and processing region $L$.

To prove Theorem~~\ref{theo_asympt_m_scheme} it remains to extend this result to $\mathcal{A}_\mathcal{S}(N) = \overline{\bigcup_l \Weyl_l(N)}^{\|\cdot\|}$ via an abstract argument.

\begin{proof}[Proof of Theorem~\ref{theo_asympt_m_scheme}]
We have shown that the set of all elements possessing asymptotic measurement schemes contains $\bigcup_k \Weyl_k(N)$, and since it is closed, also contains $\overline{\bigcup_k \Weyl_k(N)}^{\tau}=
\overline{\bigcup_k \Weyl_k(N)}^{\|\cdot\|}=\mathcal{A}_\mathcal{S}(N)$, using the fact that $\tau$ is weaker than $\tau_{\|\cdot\|}$. In particular, for every $A \in \mathcal{A}_\mathcal{S}(N)$ there exists an asymptotic measurement scheme $\qty(\mathcal{P}_\alpha, \varepsilon_{\alpha, \sigma_\alpha}, B_\alpha)_\alpha$ with coupling in $N$ and processing region $L$. If $A$ is Hermitian, then $\qty(\mathcal{P}_\alpha, \varepsilon_{\alpha, \sigma_\alpha}, \frac{1}{2}\qty(B_\alpha + B_\alpha^*))_\alpha$ is a \emph{Hermitian} $\tau$-asymptotic measurement scheme for $A$ with same coupling and processing region, see Remark~\ref{real_part_ams} below Definition~\ref{def_ams} and the fact that $\tau$ is $*$-compatible, see Lemma~\ref{lem_W_tau_cont}.
\end{proof}

\paragraph*{Remark:} This proof is the only place in which we utilize Lemma~\ref{lem_closed}. In particular, it is the only place where we need \emph{nets} rather than \emph{sequences} of measurement schemes.

\section{Physical interpretation of the approximation procedure}\label{sec_variance}

Above we saw that a Hermitian asymptotic measurement scheme exists for any system observable, both in the field algebra case (Thm.~\ref{theo_existence_of_asymp_measurement_schemes_field_algebra}) and the Weyl algebra case (Thm.~\ref{theo_asympt_m_scheme}). In this section we discuss the physical interpretation of this scheme, specifically the sequence of measurements one must perform as we scale $\lambda\rightarrow 0$.

Operationally, one first identifies a desired local system observable. Then, an appropriate probe observable is determined, and a coupling is tuned to suit these choices. Recall from above that it is possible to give tight bounds on the localization of the probe observable depending on the localization of the desired system observable. To increase the accuracy of the measurement process a ``scaling procedure'' is used, in which the coupling strength is decreased while the probe observable is upscaled. To gain some physical intuition for this scaling and how it relates to an experiment let us go through a specific example.

Suppose that our aim is to use a probe to measure the expectation value $\omega (\varphi_\mathcal{S}(f))$ of smeared field $\varphi_\mathcal{S}(f)$ in the system state $\omega$, assuming for the sake of the argument that this expectation value is nonvanishing. Focussing on the field algebra case (\Sect\ref{sec:asymp_m_schemes_field_algebra}), the aim is to extract $\omega (\varphi_\mathcal{S}(f))$ using the asymptotic measurement scheme in \Sect\ref{sec:field_alg_ams}
with a single probe field, and some coupling parameter $\lambda >0$. Following \Sect\ref{sec:field_alg_ams} the first step is to find some probe test function $h$ and coupling function $\rho$ (with the desired restrictions on their supports) such that $f$ is equivalent to $-\rho E_P^- h$, i.e. $[f]_S=[-\rho E_P^- h]_S$.

For simplicity we will assume the probe state $\sigma$ has a vanishing one-point function, i.e., $\sigma(\varphi_\mathcal{P}(f))\equiv 0$. With this assumption the probe observable we need to measure is simply the smeared field $\varphi_\mathcal{P}(\lambda^{-1}h) = \lambda^{-1}\varphi_\mathcal{P}(h)$ (see Eq.~\eqref{eq:field_scheme} in the case $k=1$). Further, since we are only using a single probe field the measurement scheme is second order (recall Theorem~\ref{theo_faster_approx_fields} and the definition of the order in~\eqref{eq_def_order_of_m_scheme}), and thus the expectation value of the corresponding coupled observable is
\begin{align}\label{eq:bias}
    \omega \left( \varepsilon_{\sigma}^{\varphi , \lambda R} (\varphi_{\mathcal{P}}(\lambda^{-1} h)) \right) & = (\omega\otimes\sigma )\left( \Theta_{\lambda} (\openone \otimes \varphi_\mathcal{P} (\lambda^{-1}h)) \right) \nonumber
    \\
    & = \omega(\varphi_\mathcal{S}(f))+\mathcal{O}(\lambda^2) \; .
\end{align}
More details of this calculation can be found in Appendix~\ref{sec_appendix_faster_approx}, e.g. equation~\eqref{eq_base_case_lambda_order_2}.

Clearly this expectation value gets closer to $\omega(\varphi_\mathcal{S}(f))$ as the coupling is turned off, i.e. as $\lambda\rightarrow 0$, which means that we can get more and more accurate readings on $\omega(\varphi_\mathcal{S}(f))$ by reducing the coupling between the probe and the main system. At the same time, however, if the `effort' associated with a measurement of $\varphi_{\mathcal{P}}(h)$ is encoded via some seminorm $\mathrm{eff}$ with $\mathrm{eff}(\varphi_{\mathcal{P}}(h))\neq 0$, then we have to put in more and more `effort' as we scale down $\lambda$. This follows as $\mathrm{eff}(\varphi_\mathcal{P}(\lambda^{-1}h)) = \lambda^{-1}\mathrm{eff}(\varphi_\mathcal{P}(h))$ diverges as $\lambda\rightarrow 0$ (\textit{cf.}~\eqref{eq_eff_diverges_for_probe} for the classical case). At the end of this section we will discuss the seminorm from Eq.~\eqref{eq_eff_sigma} as an example.

Physically, we can understand these measurements of the observables $\varphi_\mathcal{P}(\lambda^{-1}h)$ at different $\lambda$ as that of a fixed probe observable, $\varphi_\mathcal{P}(h)$, but with the additional multiplication of any measurement outcome by $\lambda^{-1}$. That is, for different values of $\lambda$ we employ the same measuring device,\footnote{Throughout the following discussion, we assume that this device returns measurement results exactly in line with the predicted statistics of $\varphi_\mathcal{P}(h)$, i.e., we neglect the fact that real devices inevitably bin results into discrete intervals.} and use it on the probe in exactly the same way. What changes is the way we process the data from the probe measurements, and, importantly, how the probe is coupled to the main system. The latter could be done in practice by tuning some other field, e.g. an electric field, that mediates the interaction between the probe and the main system.

To get a better understanding of the additional effort required as $\lambda$ is tuned down to $0$, consider a sequence of experiments running over many values of the coupling parameter $\lambda$, and for each $\lambda$ consider an ensemble of $N$ copies of the system-probe setup. Later, we will allow $N$ to depend on $\lambda$. On each copy we make a single measurement of $\varphi_\mathcal{P}(h)$ and get numerical outcomes $\mathtt{x}_{\lambda,i}$ ($i=1,...,N$) say. For each outcome we then multiply the result by $\lambda^{-1}$ as just described. Call this value $\mathtt{y}_{\lambda,i}=\mathtt{x}_{\lambda,i}/\lambda$. 

According to QFT, the outcomes $\mathtt{y}_{\lambda,i}$ are distributed 
in line with a probability distribution
determined by the observable 
$\Theta_{\lambda} (\openone \otimes \varphi_\mathcal{P} (\lambda^{-1}h))$ and the state $\omega\otimes\sigma$. Let us denote the corresponding random variable by $\mathtt{Y}_{\lambda}$ and the expectation value at coupling $\lambda$ by $\mathbb{E}_\lambda$. The \emph{observed sample mean} 
$\frac{1}{N}\sum_{i=1}^N \mathtt{y}_{\lambda,i}$ is
distributed according to the random variable $\mathtt{Y}^{(N)}_\lambda$ which is the average of $N$ independent and identically distributed copies of $\mathtt{Y}_{\lambda}$. In particular, we have
\begin{equation}\mathbb{E}_{\lambda}(\mathtt{Y}_{\lambda}^{(N)})=\mathbb{E}_{\lambda}(\mathtt{Y}_{\lambda}) =
(\omega\otimes\sigma )\left( \Theta_{\lambda} (\openone \otimes \varphi_\mathcal{P} (\lambda^{-1}h)) \right)=
\omega(\varphi_\mathcal{S}(f))+\mathcal{O}(\lambda^2)
\end{equation}
in which the first equality is a probabilistic statement and the second is predicted by QFT,
while the third comes from~\eqref{eq:bias}.
Thus the sample mean provides an estimate of $\omega(\varphi_\mathcal{S}(f))$ for sufficiently small $\lambda$, because $\mathbb{E}_{\lambda}(\mathtt{Y}_{\lambda}^{(N)})$ tends to $\omega(\varphi_\mathcal{S}(f))$ as $\lambda\to 0$, and because the $\mathcal{O}(\lambda^2)$ error is independent of $N$, the equation remains valid if $N$ is allowed to depend on $\lambda$. 

One measure of the accuracy of this estimate is the variance
\begin{equation}
    \text{Var}_{\lambda}(\mathtt{Y}_{\lambda}^{(N)}) =
   \frac{\text{Var}_{\lambda}(\mathtt{Y}_{\lambda})}{N}=
    \frac{\text{Var}_{\omega \otimes \sigma}\left(\Theta_{\lambda} (\openone \otimes \varphi_\mathcal{P} (\lambda^{-1}h))\right)}{N} \; ,
\end{equation}
where the first equality follows from the Bienaym\'{e} formula~\cite{Loeve1977}, and the fact that $\mathtt{Y}_{\lambda}^{(N)}$ is the average of independent and identically distributed random variables.

Now, $\Theta_\lambda(\openone \otimes \varphi_\mathcal{P}(\lambda^{-1}h)) = \varphi_\mathcal{S}(f_\lambda) \otimes \openone + \openone\otimes \varphi_\mathcal{P}(g_\lambda)$, where $f_{\lambda}$ and $g_{\lambda}$ are given explicitly in~\eqref{eq:convenient}. From this it is easy to see that $\text{Var}_{\omega \otimes \sigma}(\Theta_\lambda(\openone \otimes \varphi_\mathcal{P}(\lambda^{-1}h))) = \text{Var}_{\omega}(\varphi_\mathcal{S}(f_\lambda)) + \text{Var}_{\sigma}(\varphi_\mathcal{P}(g_\lambda))$, and given that $f_\lambda = f + \mathcal{O}(\lambda^2)$ (\textit{cf.}~\eqref{eq_base_case_lambda_order_2} in Appendix~\ref{sec_appendix_faster_approx}) we then have that $\text{Var}_{\omega}(\varphi_\mathcal{S}(f_\lambda)) = \text{Var}_{\omega}(\varphi_\mathcal{S}(f)) + \mathcal{O}(\lambda^2)$. Since $g_\lambda = \lambda^{-1}h + \lambda R E_S^-RE_P^- h + \mathcal{O}(\lambda^3)$, as shown in Eq.~\eqref{eq:glambda_expansion} in the case $k=1$, one can similarly verify that
\begin{equation}
    \text{Var}_{\sigma}(\varphi_\mathcal{P}(g_\lambda)) =  \lambda^{-2} \text{Var}_{\sigma}(\varphi_\mathcal{P}(h)) + \sigma( \lbrace  \varphi_\mathcal{P}(h) , \varphi_\mathcal{P}(RE_S^-RE_P^-h)  \rbrace ) + \mathcal{O}(\lambda^2)  \; ,
\end{equation}
where $\lbrace \cdot , \cdot \rbrace$ denotes an anti-commutator, and we have used the fact that one-point functions vanish for the probe state $\sigma$.  Putting both the system and probe variances together we have
\begin{align}\label{eq:VarYNlambda}
    \text{Var}_{\lambda}(\mathtt{Y}^{(N)}_\lambda) = &  \frac{\text{Var}_{\sigma}(\varphi_\mathcal{P}(h))}{\lambda^2 N} \nonumber
    \\
    & + \frac{\text{Var}_{\omega}(\varphi_\mathcal{S}(f)) + \sigma( \lbrace  \varphi_\mathcal{P}(h) , \varphi_\mathcal{P}(RE_S^-RE_P^-h)  \rbrace )}{N} \nonumber
    \\
    & + \frac{\mathcal{O}(\lambda^2)}{N} \; ,
\end{align}
In the foregoing expressions, all the $\mathcal{O}(\lambda^2)$ errors are independent of $N$. We can now see that, while the sample mean, $\mathtt{Y}^{(N)}_\lambda$, always tends to $\omega(\varphi_\mathcal{S}(f))$ in expectation as we tune down the coupling $\lambda$, we have to simultaneously increase the number of trials $N$ in order to keep the accuracy of the estimate the same, as measured by its variance. This may be further quantified as follows, using elementary statistics.

In the terminology of classical estimation theory~\cite{Casella2002}, each $\mathtt{Y}^{(N)}_\lambda$ is a \emph{biased estimator} for  $\omega(\varphi_\mathcal{S}(f))$. Nevertheless, it may serve as the midpoint of a \emph{confidence interval} for $\omega(\varphi_\mathcal{S}(f))$. 
As a reminder, given $\epsilon>0$ and $0<\delta < 1$, we say that $\qty(\mathtt{Y}^{(N)}_\lambda-\epsilon, \mathtt{Y}^{(N)}_\lambda+ \epsilon)$ is a confidence interval for $\omega(\varphi_\mathcal{S}(f))$ with confidence coefficient $1-\delta$ 
(or a $100(1-\delta)\%$ confidence interval) if
\begin{equation}  
\Prob_{\lambda} \qty(\qty|\mathtt{Y}^{(N)}_\lambda -\omega(\varphi_\mathcal{S}(f))| < \epsilon) \geq 1- \delta.
\end{equation}

\begin{theo}\label{thm:effort}
For every $\epsilon>0$ and $0 < \delta < 1$ there is a $\lambda_0 >0$ and for every $0 < \lambda < \lambda_0$ there is $N_\lambda \in \mathbb{N}$ such that $\qty(\mathtt{Y}^{(N_\lambda)}_\lambda-\epsilon, \mathtt{Y}^{(N_\lambda)}_\lambda+ \epsilon)$ is a confidence interval for $\omega(\varphi_\mathcal{S}(f))$ with confidence coefficient $1-\delta$. For small $\epsilon$ we have
\begin{equation}
    N_\lambda \gtrsim \frac{\mathrm{Var}_{\sigma}(\varphi_\mathcal{P}(h))/\lambda^2}{\delta (\epsilon - C\lambda^2)^2} 
    =\frac{\mathrm{eff}_{\sigma}(\varphi_{\mathcal{P}}(h_\lambda))^2}{\delta\epsilon^2}(1-C\lambda^2/\epsilon)^{-2}
    \label{eq_condition_N}
\end{equation}
for a constant $C$.
\end{theo}

In the proof we see how a choice of $\lambda_0$ is used to control the bias of $\mathtt{Y}^{(N)}_\lambda$. Furthermore,~\eqref{eq_condition_N} then directly relates the physical resources required for the measurement, specifically the number of trials $N_\lambda$, to 
\begin{enumerate}
    \item the given ``tolerance'' $\epsilon$ and ``accuracy'' $\delta$ (also called the significance level),
    \item the measure of effort for the \emph{asymptotic} measurement scheme introduced in~\eqref{eq_eff_sigma}, (which quantifies the variance of the estimator), and
    \item the bias of the estimator.
\end{enumerate}

Before we turn to the proof we discuss the interpretation of Theorem~\ref{thm:effort}. Suppose one is given a device for measuring $\varphi_{\mathcal{P}}(h)$ and the task is to determine $\omega(\varphi_\mathcal{S}(f))$ up to a maximal absolute error of $\epsilon > 0$ with a probability of $1-\delta$. Then Theorem~\ref{thm:effort} tells us that there is a $\lambda >0$ and an $N_\lambda$ (determined by~\eqref{eq_condition_N}) such that, with a probability of $1-\delta$, the observed sample mean $\frac{1}{N_\lambda} \sum_{i=1}^{N_\lambda} \mathtt{y}_{\lambda,i}$ of $N_\lambda$ measurements at coupling parameter $\lambda$ deviates by at most $\epsilon$ from $\omega(\varphi_\mathcal{S}(f))$. 

\begin{proof}
We first note that, since $\mathbb{E}_{\lambda}({\mathtt{Y}^{(N)}_\lambda})=\omega(\varphi_\mathcal{S}(f))+\mathcal{O}(\lambda^2)$, there is a $\tilde{\lambda}_0>0$ and a constant $C$ such that $|\omega(\varphi_\mathcal{S}(f))-\mathbb{E}_{\lambda}({\mathtt{Y}^{(N)}_\lambda})|
\le C\lambda^2$ for all $0<\lambda< {\tilde{\lambda}_0}$. 
If $|{\mathtt{Y}^{(N)}_\lambda}-\omega(\varphi_\mathcal{S}(f))|\ge \epsilon$, then the reverse triangle inequality gives
\begin{equation}
    |{\mathtt{Y}^{(N)}_\lambda}- \mathbb{E}_{\lambda}({\mathtt{Y}^{(N)}_\lambda})|\ge
    |{\mathtt{Y}^{(N)}_\lambda}-\omega(\varphi_\mathcal{S}(f))|-| \omega(\varphi_\mathcal{S}(f))-\mathbb{E}_{\lambda}({\mathtt{Y}^N})|
    \ge \epsilon - C\lambda^2 >0 \; ,
\end{equation}
for $0<\lambda<\lambda_0$ (independent of $N$), where $\lambda_0 \leq \mathrm{min} \qty{\tilde{\lambda}_0, \sqrt{\epsilon/C}}$. Thus 
\begin{equation}
    \Prob_{\lambda}\left(|{\mathtt{Y}^{(N)}_\lambda}-\omega(\varphi_\mathcal{S}(f))|\ge \epsilon\right) \le
    \Prob_{\lambda}\left(|{\mathtt{Y}^{(N)}_\lambda}- \mathbb{E}_{\lambda}({\mathtt{Y}^{(N)}_\lambda})|
    \ge \epsilon - C\lambda^2\right)\le \frac{ \text{Var}_{\lambda}({\mathtt{Y}_{\lambda}^{(N)}})}{(\epsilon - C\lambda^2)^2} \; ,
\end{equation}
using Chebyshev's inequality. Following this, $N$ must be chosen sufficiently large such that the \emph{variance} of the estimator $\mathtt{Y}^{(N)}_\lambda$ is controlled by
\begin{equation}
     \frac{ \text{Var}_{\lambda}({\mathtt{Y}^{(N)}_\lambda})}{(\epsilon - C\lambda^2)^2} <\delta \; .
\end{equation}
From~\eqref{eq:VarYNlambda} it is then clear that this condition can always be satisfied. Quantitatively, we require
\begin{equation}
    N_{\lambda}> \frac{1}{\delta (\epsilon - C\lambda^2)^2}\left(
    \frac{\text{Var}_{\sigma}(\varphi_\mathcal{P}(h))}{\lambda^2} + \mathcal{O}(1)
    \right) 
\end{equation}
as $\lambda\to 0$.
At least for small $\epsilon$ (for which $\lambda_{0}$ and consequently $\lambda$ are small) the $\mathcal{O}(1)$ terms may be neglected and we find~\eqref{eq_condition_N}.
\end{proof}

\section{Discussion and Outlook}\label{sec:discussion}

Even though a satisfactory description of the full quantum measurement process is not yet available, considerable insight can be obtained by analysing the measurement chain by which information about a system can be inferred from measurements of a probe. The implementation of this idea for relativistic QFT~\cite{fewster2018quantum} has produced an operationally well-motivated framework of local measurement schemes (FV framework) that are consistent with causality~\cite{bostelmann2020impossible}. In particular, these measurements are free of Sorkin's superluminal signaling problem~\cite{sorkin1993impossible}. In the present paper we have set out to address the question whether every local observable of linear real scalar fields\footnote{In fact we do not expect any obstruction to extending our results to a collection of multiple scalar fields.} on a globally hyperbolic spacetime $M$ can be measured in the FV framework.

To that end we first discussed a proof of principle, which however only fully complies with the requirements of the FV framework in the case where $M$ has compact Cauchy surfaces. In order to address the general question we introduced the notion of a Hermitian \emph{asymptotic measurement scheme} for a system observable $A$, given by a collection of bona fide Hermitian measurement schemes whose induced observables converge to $A$. It is the main result of the present work that every local observable $A$ of a linear scalar QFT on $M$ admits a Hermitian asymptotic measurement scheme, which requires an observer to have control over their probe theories only in a reasonably small spacetime region depending on the localization of $A$. We achieve this by first analysing asymptotic measurement schemes for the underlying classical scalar field theory, which then straight forwardly lead to (Hermitian) asymptotic measurement schemes for the quantum field theory (described using both the field algebra and the Weyl algebra). Furthermore, we have given a precise mathematical abstraction of the \emph{effort} associated to an asymptotic measurement scheme. In an example this is given by the square root of the variance of the probe observables in the probe preparation state and hence also bears a clear operational meaning as discussed in \Sect\ref{sec_variance}. The efficiency of an asymptotic measurement scheme can then be quantified by comparison of the rate of convergence of the induced system observables to the rate of divergence of the effort, leading to the notion of \emph{order} of a quantum asymptotic measurement scheme (for smeared fields). In particular, by using multiple probe fields, we have shown that every smeared system field admits an asymptotic measurement scheme of any even order.\\

While the use of relativistic quantum fields as probes themselves stems from~\cite{fewster2018quantum}, non-relativistic structures such as the two-level systems employed by Unruh and Wald~\cite{UnruhWald1984}, based on earlier work~\cite{Unruh1976,DeWitt1979}, have been studied in the literature for a long time and are an important tool in the field of relativistic quantum information (RQI).
However, until very recently, previous work in RQI did not focus on induced system observables. A noteworthy exception was~\cite{Smith2019}; subsequent to~\cite{fewster2018quantum,bostelmann2020impossible} there has been increased interest in this direction and a 
treatment of non-relativistic probes in a similar spirit to~\cite{fewster2018quantum} can be found in~\cite{pologomez2021,deRamon2021}. It should be mentioned that the coupling of non-relativistic structures to quantum fields typically requires either a singular or a non-local coupling, of which the latter is in conflict with causality, while measurements in the FV framework are causal. (See also~\cite{Borsten2021} and~\cite{jubb2021causal} for an analysis of general causal state-updates, of which the resulting updates of FV measurement schemes are (at least) a special case.)\\

The last point together with the result of the present paper that every local observable can be measured in the FV framework (via a Hermitian asymptotic measurement scheme) indeed reinforces the interpretation of generic Hermitian elements of local algebras as \emph{observables} in contrast to \emph{local operations}, which is an important consequence of our investigation.\\

An immediate question raised by our results is whether they extend to the algebra generated by Wick powers. In the expectation of merely mild technical effort we have left this task together with the possible extension to theories other than scalar ones for future work. Beyond specific quantum field theories, it would be very interesting to understand whether there is a general result, akin to the Stinespring theorem, that would establish the (asymptotic) measureability of arbitrary local observables in general QFTs by locally coupled probes.

\section*{Acknowledgment}
It is a pleasure to thank Henning Bostelmann for many useful discussions, and also Fay Dowker and Rafael Sorkin for insightful comments and questions. We also thank the anonymous referees for helpful comments and remarks. The work of M.~H.~R. was supported by a Mathematics Excellence Programme Studentship awarded by the Department of Mathematics at the University of York, and the work of I. J. was supported by a DIAS Scholarship. 

\appendix

\section{The scattering operator of \Sect\ref{sec:proof_of_p}}
\label{sec_appendix_proof_of_p}

The scattering operator is easily determined (from~\eqref{eq:scattering}), noting that the `in' and `out' regions are $M^\pm = I^\pm(\Sigma^\pm)$. For $F\in C_c^\infty(M^+;\mathbb{C})$, one has $\Theta\Phi(F)=\Phi(\tilde{F})$ where $\tilde{F}\in C_c^\infty(M^-;\mathbb{C})$ obeys $E_Q\tilde{F}=E_QF$, or equivalently $E_P e^{i\chi} F= E_P e^{i\chi} \tilde{F}$. 
Because $\chi$ vanishes on the support of $F$ and takes the value $\pi/2$ on the support of $\tilde{F}$, the requirement simplifies to
\begin{equation}
    E_P \tilde{F} = -i E_P F
\end{equation}
which implies that $\tilde{F} = -iF + PH$ for some $H\in C_c^\infty(M;\mathbb{C})$. It follows that $\Theta\Phi(F) = -i\Phi(F)$ for all $F\in C_c^\infty(M^+;\mathbb{C})$  and hence
for all $F\in C_c^\infty(M;\mathbb{C})=C_c^\infty(M^+;\mathbb{C})+PC_c^\infty(M;\mathbb{C})$ using the classical timeslice property.

\section{Proof of Lemma~\ref{lem_reduction_to_slab_classical}}\label{sec_appendix_proof_theo_reduction_slab}

Lemma~\ref{lem_reduction_to_slab_classical} is based on the case $k=1$ of the following lemma. The general case is used in the proof of Theorem~\ref{theo_faster_approx} in Appendix~\ref{sec_appendix_faster_approx}.

\begin{lem}
Let $S, P_1,...,P_k$ be the system and probe equation of motion operators. For every region $N \subseteq M$ and every finite collection of test functions $f_1, ..., f_k \in C_c^\infty(N;\mathbb{R})$ there exists a precompact region $\tilde{N} \subseteq N$ and $\tilde{f}_1, ..., \tilde{f}_k \in C_c^\infty(\tilde{N};\mathbb{R})$ such that $\forall j:$
\begin{enumerate}
    \item $ [f_j]_S= [\tilde{f}_j]_S$,
    \item $ \exists \, \varphi_j \in \mathrm{Sol}_{sc}(P_j)$ such that $\varphi_j \restriction \tilde{N}$ is nowhere vanishing.
\end{enumerate}
Moreover, for any region $L$ whose domain of dependence contains $N$, there exist $h_j \in C_c^\infty(L;\mathbb{R})$ such that $\varphi_j = E_{P_j}h_j$ ($1\le j\le k$).
\label{lem_finding_slab}
\end{lem}

\begin{proof}
Since $N$ is a region, $(N,\mathtt{g}|_N)$ is a globally hyperbolic spacetime in its own right, where $\mathtt{g}$ denotes the Lorentzian metric on the spacetime manifold $M$. It should be borne in mind that $N$ might have finitely many connected components. One can then find a Cauchy surface, $\Sigma\subset N$, and a diffeomorphism $\xi:\mathbb{R}\times\Sigma\to N$ with $\xi(\{0\}\times\Sigma)=\Sigma$ and such that the pullback of the metric has the form $(\xi^*\mathtt{g})(t,\vec{x})=\Omega(t,\vec{x})^2(1\oplus -h_t(\vec{x}))$, where $t\mapsto h_t$ is a smooth family of Riemannian metrics on $\Sigma$, and $\Omega\in C^\infty(\mathbb{R}\times\Sigma)$ is nonvanishing~\cite{bernal2003smooth}. In particular, every curve of the form $t\mapsto (t,\vec{x})$, for constant $\vec{x}\in\ \Sigma$, is timelike and (without loss of generality) future-directed. The associated Cauchy time function $\mathtt{t}\in C^\infty(N;\mathbb{R})$ is defined so that $\mathtt{t}(\xi(t,\vec{x}))=t$ for all $(t,\vec{x})\in\mathbb{R}\times\Sigma$. Additionally note that any subset of the form $\xi (I\times \Sigma )\subset N$, for some open interval $I\subset \mathbb{R}$, is causally convex. In what follows we identify $N$ with $\mathbb{R}\times \Sigma$ for convenience, and similarly $\xi(A \times \Sigma)$ with $A \times \Sigma$ for any subset $A\subseteq \mathbb{R}$.

All the solutions sourced by $f_1, ..., f_k \in C_c^\infty(N)$ have compactly supported initial data on $\Sigma$. Let $B \subseteq \Sigma$ be an open precompact subset with finitely many connected components, $B_i$, such that every initial data has support in $B$ (via compact exhaustion). Let $\varphi_j$ be the solution to $P_j$ with compactly supported Cauchy data $(\phi,0)$ (i.e., vanishing normal derivative) independent of $j$. Here $\phi \in C_c^\infty(\Sigma)$ is some smoothed characteristic function of $\overline{B}$, i.e. $\phi \restriction \overline{B}=1$. As each $\varphi_j$ is smooth and set to $1$ on $B$, there is some open neighbourhood of $B$ of the form $I\times B$, for some open interval $I\subset\mathbb{R}$ containing $0$, within which every $\varphi_j$ is non-vanishing. If this were not the case one could find a sequence of zeros, $z_n\in (-1/n,1/n)\times B$, for the continuous function $\prod_{j=1}^k \varphi_j$, which must therefore vanish somewhere in $\overline{B}$ by compactness and continuity, thus producing a contradiction.

We then define $\tilde{N}:= D(B) \cap (I \times B)$, which is open and, moreover, precompact in $M$ (as $ \overline{\tilde{N}}\subset\overline{I \times B}$, which is a compact set in $M$). Note that all the $\varphi_j$ are non-vanishing on $\tilde{N}$, as desired in point \textit{(ii)} of Lemma~\ref{lem_finding_slab}. We next need to show that $\tilde{N}$ is also a region, for which it is necessary to check that $\tilde{N}$ is \textit{a)} causally convex, and \textit{b)} that it has only finitely many connected components.

\noindent\textit{a)} Causal convexity can be seen as follows. Due to the form of the metric, $D(\{0\}\times B)\subset \mathbb{R}\times B$, since any $(t',\vec{x})$ with $\vec{x}\notin B$ lies on an inextendible timelike curve $t\mapsto(t,\vec{x})$ that does not intersect $\{0\}\times B$. Abusing notation by writing $D(\{0\}\times B)=D(B)$, we then have $D(B)=D(B)\cap (\mathbb{R}\times B)$, and
\begin{align} 
D(B)\cap (I\times B) & =  D(B)\cap(\mathbb{R}\times B) \cap (I\times B) \nonumber
\\
& = D(B)\cap (\mathbb{R}\times B)\cap (I\times\Sigma) \nonumber
\\
& = D(B)\cap (I\times\Sigma) \; ,
\end{align}
using $(\mathbb{R}\times B) \cap (I\times B)=I\times B=(\mathbb{R}\times B)\cap (I\times\Sigma)$. The rhs of the last line is the intersection of two causally convex sets, and is therefore casually convex, which immediately implies that $\tilde{N} = D(B)\cap (I\times B)$ is as well.

\noindent\textit{b)} To see that $B$ has finitely many connected components we write it as the union of its connected components $B_j$. Since the $B_i$'s are mutually spacelike, $D(B)= \bigcup_j D(B_j)$, and hence it suffices to show that each $D(B_j) \cap (I\times\Sigma)$ is connected. For every point $(t',\vec{x})$ in the intersection (taking $t'>0$ without loss of generality), the curve $[0,t')\ni t \mapsto (t,\vec{x})$ is also fully contained in the intersection. So every point is path-connected to $B_j$, but $B_j$ is itself connected and hence also path connected, so $D(B_j) \cap (I\times\Sigma)$ is connected.

With the precompact region $\tilde{N}$ in hand we can now find $\tilde{f}_1, ..., \tilde{f}_k \in C_c^\infty(\tilde{N})$ with the stated properties as follows. Explicitly, consider any (necessarily precompact) open neighbourhood $O\subset \tilde{N}$ of $B$, and partition the solution $\psi_j = E_S f_j$ as $\psi_j = \psi_j^+ + \psi_j^-$ with $\supp \psi_j^\pm\subset J^\pm(O)$. Then $\psi^{\pm}$ has past/future-compact support (in $M$) so that $\tilde{f}_j:=S\psi_j^-=-S\psi_j^+$ is compactly supported in $J^+(O)\cap J^-(O)$. The latter is contained in $\tilde{N}$ by causal convexity, and hence $\tilde{f}_j \in C_c^{\infty}(\tilde{N})$. Finally,
\begin{equation}
    E_S \tilde{f}_j = E_S S\psi_j^-=E_S^- S\psi_j^-- E_S^+ (-S\psi_j^+)=\psi_j^-+\psi_j^-=\psi_j=E_S f_j \; ,
\end{equation}
and hence $[\tilde{f}_j]_S = [f_j]_S$ as desired in point \textit{(i)} of Lemma~\ref{lem_finding_slab}.

Finally, if $L$ is any region in $M$ whose domain of dependence contains $N$, the solutions $\varphi_j$ introduced above can be written in the form $\varphi_j=E_{P_j}h_j$ for suitable $h_j\in C^\infty_c(L;\mathbb{R})$.
\end{proof}

We are now ready to prove Lemma~\ref{lem_reduction_to_slab_classical}.

\begin{proof}[Proof of Lemma~\ref{lem_reduction_to_slab_classical}.]
The first part immediately follows from Lemma~\ref{lem_finding_slab} for $k=1$. It remains to show that $ \exists \rho \in C_c^\infty(\tilde{N};\mathbb{R}) , \; \exists h \in C_c^\infty(L;\mathbb{R}):\; \tilde{f}=- R E_{P}^- h$.

By Lemma~\ref{lem_finding_slab} we can find $\varphi$ in $\mathrm{Sol}_{sc}(P)$, and a region $\tilde{N}$ such that $\varphi$ is non-vanishing on $\tilde{N}$. Thus $\rho:=-\tilde{f}/\varphi$ (or more precisely the extension by zero of this function from $\tilde{N}$ to $M$) defines a smooth function with support equal to that of $\tilde{f}$. As a result, $\tilde{f} = - R \varphi$. Nothing so far has depended on the choice of $L$. To show that $R\varphi$ can be written as $R E^-_{P} h$ we first use Lemma~\ref{lem_finding_slab} to find $h \in C_c^\infty(L;\mathbb{R})$ such that $\varphi=E_{P} h$. But since the support of $h$ is in $ L\subseteq M^+$, we see that $R E_{P}^+ h=0$ and hence $R E_{P} h = R E_{P}^- h$, which finishes the proof.
\end{proof}

\section{Proofs of Theorem~\ref{theo_faster_approx} and Theorem~\ref{theo_faster_approx_fields}}\label{sec_appendix_faster_approx}

We start with some definitions. First, $\mathcal{L}_b(X,Y)$ will denote the space of continuous maps between locally convex topological spaces $X$ and $Y$, equipped with the topology of bounded convergence~\cite{Treves:TVS} (which generalises the topology of convergence in operator norm for maps between normed spaces). In the case $Y=X$ we write simply $\mathcal{L}_b(X)$.
Second, the expression $X(\lambda) = \mathcal{O}(\lambda^s)$ will mean that $X(\lambda)=\lambda^{ s}U(\lambda)$ where $U$ is analytic with respect to a stated topology. Finally, recall that a \emph{Green hyperbolic operator} on $C^\infty(M;\mathbb{R}^k)$
is a partial differential operator which has advanced and retarded Green operators, as does its formal adjoint~\cite{Baer:2015}. This includes the normally hyperbolic second order operators discussed in the text but it is convenient to take a more general standpoint here (in fact, one could go even further and consider differential operators between vector bundles).

Proceeding with some general observations: Let $T_0$ and $T_1$ and $T_2$ be Green hyperbolic operators on $C^\infty(M;\mathbb{R}^{\ell +1})$ for some $\ell$, with
corresponding Green operators $E^\pm_{T_j}$ ($j=0,1,2$). 
Suppose further that these operators differ from one another only within a compact region $K$, so that $(T_j-T_k)f$ is supported in $K$ for every $f\in C^\infty(M;\mathbb{R}^{\ell +1})$.
For $j=1,2$, set
\begin{equation}
    \theta_j = I - (T_j-T_0) E_{T_j}^-
    \label{eq_def_theta_j}
\end{equation}
which (as in~\eqref{eq:unvartheta_def}) is the classical scattering operator for the dynamics of $T_j$ relative to that of $T_0$ (modulo passing to equivalence classes). We compute
\begin{align}
    \theta_2-\theta_1 &= (T_1-T_0)E_{T_1}^- - (T_2-T_0)E_{T_2}^-\notag\\
    &= T_0(E_{T_2}^--E_{T_1}^-)\notag\\
    &= T_0E_{T_1}^-(T_1-T_2)E_{T_2}^-\notag\\
    &= (I-(T_1-T_0)E_{T_1}^-)(T_1-T_2)E_{T_0}^-(I-(T_2-T_0)E_{T_2}^-),
\end{align}
where we have used the identity $E_{T_i}^- = E_{T_j}^-(I- (T_i-T_j)E_{T_i}^-)$ three times. 

Now let $T_1$ and $T_2$ depend analytically on a parameter $\lambda$ with $T_1(0)=T_2(0)=T_0$,
and $T_1(\lambda)-T_2(\lambda)=\mathcal{O}(\lambda^{2k})$ for some $k\ge 1$, with respect to the topology of $\mathcal{L}_b(C^\infty(M;{\mathbb{C}^{{ \ell}+1}}))$, where
$C^\infty(M;{\mathbb{C}^{{ \ell}+1}})$ is given its standard Fr\'echet topology.
As the $T_j(\lambda)$ agree outside $K$, the differences are also analytic in the $\mathcal{L}_b(C^\infty(M;{\mathbb{C}^{{ \ell}+1}}),C_c^\infty(M;{\mathbb{C}^{{ \ell}+1}}))$ where
$C_c^\infty(M;{\mathbb{C}^{{ \ell}+1}})$ has its standard LF space topology. (See~\cite{Treves:TVS} for the definition of the topologies involved.) In the same way, $T_j(\lambda)-T_0=\lambda T'_j(0) + \mathcal{O}(\lambda^2)$ for $j=1,2$.
Assume also that $E_{T_j(\lambda)}^-$ are analytic in $\lambda$ with respect to the topology $\mathcal{L}_b(C_c^\infty(M;{\mathbb{C}^{{ \ell}+1}}),
C^\infty(M;{\mathbb{C}^{{ \ell}+1}}))$, whereupon the composition $(T_l(\lambda)-T_m(\lambda))E_{T_n(\lambda)}^-$ is analytic in $\lambda$ in the topology of $\mathcal{L}_b(C_c^\infty(M;{\mathbb{C}^{{ \ell}+1}}))$. Then
\begin{align}
\theta_2(\lambda)-\theta_1(\lambda) &= 
(T_1(\lambda)-T_2(\lambda))E_{T_0}^- 
-\lambda T_1'(0)E_{T_0}^-(T_1(\lambda)-T_2(\lambda))E_{T_0}^-
\notag\\ &\qquad 
-\lambda (T_1(\lambda)-T_2(\lambda))E_{T_0}^- T_2'(0) E_{T_0}^-
+\mathcal{O}(\lambda^{2k+2})\notag\\
&= -\lambda^{2k}VE_{T_0}^- + \lambda^{2k+1}(T_1'(0)E_{T_0}^-V E_{T_0}^- + V E_{T_0}^-T_1'(0) E_{T_0}^-)  \notag\\
&\qquad 
+\mathcal{O}(\lambda^{2k+2})
\label{eq_diff_theta_gen}
\end{align}
in $\mathcal{L}_b(C_c^\infty(M;{\mathbb{C}^{{ \ell}+1}}))$ where we write $T_2(\lambda)-T_1(\lambda)=\lambda^{2k}V+\mathcal{O}(\lambda^{2k+1})$, and note that $T_2'(0)=T_1'(0)$ by assumptions made above.\\

Let us now consider for $\ell \geq 1$ and $k \leq \ell$ the following concrete operators
\begin{equation}
    T_k(\lambda)= \begin{pmatrix} S & {R^{(k)}}(\lambda)^T \\ R^{(k)}(\lambda) & P^{\oplus \ell} \end{pmatrix},
\end{equation}
where $S$ and $P$ are formally self-adjoint Green hyperbolic operators on $C^\infty(M;\mathbb{R})$ and
\begin{equation} 
R^{(k)}(\lambda)^T:= \qty(\lambda R_1, \lambda^2 R_2,   \lambda^4 R_3,\ldots, \lambda^{2k-2} R_k,0,\ldots,0): C_c^\infty(M;\mathbb{R}^\ell) \to C_c^\infty(M;\mathbb{R})
\end{equation}where each $R_j$ is the operator of multiplication by some compactly supported function $\rho_j$ with support in $K$. The operators $T_k(\lambda)-T_k(0)$ are polynomials with smooth coefficients and therefore analytic with respect to the topology of $\mathcal{L}_b(C^\infty(M;\mathbb{C}^{{ \ell}+1}))$, which is also true of $T_k(\lambda)$ in consequence. By results in~\cite{fewster-non-local}, the Green operators are analytic in $\lambda$ with respect to the topology of
$\mathcal{L}_b(C_c^\infty(M;\mathbb{C}^{{ \ell}+1}),C^\infty(M;\mathbb{C}^{{ \ell}+1}))$.
Therefore, the assumptions needed for the calculations above are valid.

Introducing the $(\ell+1) \times (\ell+1)$ matrix units $M_{i,j}$, i.e., $\qty(M_{i,j})_{l,m}:= \delta_{i,l} \delta_{j,m}$, we may write
\begin{align}
    T_{k+1}(\lambda)-T_{k}(\lambda)&=\lambda^{2k} R_{k+1}\otimes (M_{1, k+2} + M_{k+2, 1}),\\
    T_1'(0) &= R_1\otimes (M_{1,2}+M_{2,1}),
\end{align}
and using Eq.~\eqref{eq_diff_theta_gen} and the fact that $E_{T_0}^-$ is diagonal, we compute
\begin{align}
\theta_{k+1}(\lambda)-\theta_k(\lambda) &= -\lambda^{2k} (R_{k+1}E_P^{-}\otimes M_{1,k+2} + R_{k+1}E_S^-\otimes M_{k+2,1}) \\
&\qquad +\lambda^{2k+1} (R_1 E_S^- R_{k+1}\otimes M_{2,k+2} + R_{k+1}E_S^{-}R_1\otimes M_{k+2,2})\\
&\qquad + \mathcal{O}(\lambda^{2k+2})
\end{align}
in $\mathcal{L}_b(C_c^\infty(M;{\mathbb{C}^{{ \ell}+1}}))$.
In particular, the projection onto the first component is
\begin{equation}
      \text{pr}_1\circ (\theta_{k+1}(\lambda)-\theta_k(\lambda)) = -\lambda^{2k} R_{k+1}E_P^{-}\circ \text{pr}_{k+2} + \mathcal{O}(\lambda^{2k+2}),
\end{equation}
in $\mathcal{L}_b(C_c^\infty(M;{\mathbb{C}^{{ \ell}+1}}),C_c^\infty(M;{\mathbb{C}}))$.

Turning to our application, let us recall from Lemma~\ref{lem_finding_slab}, that for every $f \in C_c^\infty(N;\mathbb{R})$ there exists $\tilde{f} \in C_c^\infty(N;\mathbb{R})$ and a solution $\varphi \in \mathrm{Sol}_{sc}(P)$ that is nonzero throughout $\supp \tilde{f}$. Further, fixing $K:=\mathrm{supp}\; \tilde{f}$ and choosing any region $L \subseteq M^+= M\setminus J^-(K)$ whose domain of dependence contains $N$, there exists $h \in C_c^\infty(L, \mathbb{R})$ such that $\varphi=E_{P}h$. It follows that $E^-_{P}h$ is nonvanishing throughout $K$, on which it agrees with $\varphi$. Hence, for
\begin{equation}
 \vec{h}_{\lambda}^T:= \qty(\lambda^{-1} h, h, ..., h) \in C_c^{\infty}(M;\mathbb{R}^\ell),
 \end{equation}
we find
\begin{equation}
 \text{pr}_1\circ (\theta_{k+1}(\lambda)-\theta_k(\lambda))\mqty(0\\ \vec{h}_\lambda) = 
 -\lambda^{2k} R_{k+1} \varphi + \mathcal{O}(\lambda^{2k+2}),
  \label{eq_pr_diff_theta}
\end{equation}
in $C_c^\infty(M, \mathbb{R})$ since by assumption $\mathrm{supp} \; \rho_{k+1} \subseteq K$.

Let us now choose the functions $\rho_k$ inductively in the following way. First, observe that $\rho_1:= - \tilde{f}/ \varphi$ is a smooth function with support in $K$ and that hence
\begin{align}\label{eq_base_case_lambda_order_2}
    \text{pr}_1\circ \theta_{1}(\lambda)\mqty(0\\ \vec{h}_\lambda)& = - R_1 \varphi + \mathcal{O}(\lambda^{2}) =\tilde{f} + \mathcal{O}(\lambda^2).
\end{align}
This is the base case $k=1$ of the inductive hypothesis that
\begin{equation}\label{eq:inductive_hyp}
 \text{pr}_1\circ \theta_{k}(\lambda)\mqty(0\\ \vec{h}_\lambda) = 
 \tilde{f} + \lambda^{2k}\mathcal{E}_k(\lambda),
 \end{equation}
where $\lambda \mapsto \mathcal{E}_k(\lambda) \in C_c^\infty(M;\mathbb{R})$ is analytic and valued in functions supported in $K$. This last property follows because $\text{supp}\,\tilde{f}=K$, together with the fact that 
the left hand side of Eq.~\eqref{eq:inductive_hyp} is supported in $K$, by 
inspection of Eq.~\eqref{eq_def_theta_j} and because the first component of $(0,\vec{h}_\lambda)^T$ is zero. Next, suppose that the hypothesis has been established for all $1 \leq j \leq k$ for some $1\le k\le \ell-1$. Then, by Eq.~\eqref{eq_pr_diff_theta}
\begin{equation}
 \text{pr}_1\circ \theta_{k+1}(\lambda)\mqty(0\\ \vec{h}_\lambda) = 
 \tilde{f} +\lambda^{2k}\mathcal{E}_k(\lambda)  -\lambda^{2k} R_{k+1}\varphi + \mathcal{O}(\lambda^{2k+2})
 \end{equation}
so the choice 
\begin{equation}
    \rho_{k+1} = \mathcal{E}_k(0)/\varphi
\end{equation}
ensures that~\eqref{eq:inductive_hyp} holds with $k$ replaced by $k+1$ and that $\mathrm{supp}\; \rho_{k+1} \subseteq K$. Thus the statement~\eqref{eq:inductive_hyp} holds for all $1\le k\le \ell$, in particular
\begin{equation}
 \text{pr}_1\circ \theta_{\ell}(\lambda)\mqty(0\\ \vec{h}_\lambda) = 
 \tilde{f} + \lambda^{2\ell}\mathcal{E}_\ell(\lambda).
 \end{equation} 
Let us now turn to the proofs.

\begin{proof}[Proof of Theorem~\ref{theo_faster_approx}]
We note that $T_\ell(\lambda)$ and the probe observables $\qty[\vec{h}_\lambda]_{P^{\oplus \ell}}$ give rise to a collection of classical measurement schemes
\begin{equation}
    H^{\mathrm{cl},\ell}_\lambda = \qty(\mathcal{C}_\mathcal{P}^{\oplus \ell}, \varepsilon^{\mathrm{cl}, R_\lambda}, \qty[\vec{h}_\lambda]_{P^{\oplus \ell}}),
\end{equation}
where
\begin{equation}
    \varepsilon^{\mathrm{cl}, R_\lambda}\qty(\qty[\vec{h}_\lambda]_{P^{\oplus \ell}}) = \text{pr}_1\circ \theta_{\ell}(\lambda)\mqty(0\\ \vec{h}_\lambda)=  \tilde{f} + \lambda^{2\ell}\mathcal{E}_\ell(\lambda).
\end{equation}
Obviously, $\lambda \mapsto \qty[\mathcal{E}_k(\lambda)]_S$ is $\tau^{\mathrm{cl}}$ continuous at $\lambda=0$, hence $\qty(H^{\mathrm{cl},\ell}_\lambda)_{\lambda >0}$ is an asymptotic measurement scheme for $\qty[f]_S =\qty[\tilde{f}]_S$ with coupling in $N \supseteq \mathrm{supp}\; \tilde{f}$ and processing region $L$.

For any choice of seminorm $\mathrm{eff}$ such that $\mathrm{eff}\qty([\vec{h}_\lambda]_{P^{\oplus\ell}}) \neq 0$, we have that
\begin{equation}
    \mathrm{eff}([\vec{h}_\lambda]_{P^{\oplus\ell}}) = \lambda^{-1} \mathrm{eff}([(h, \lambda h, ..., \lambda h)^T]_{P^{\oplus\ell}}),
\end{equation}
where $\lambda \mapsto \mathrm{eff}([(h, \lambda h, ..., \lambda h)^T]_{P^{\oplus\ell}}) \in \mathbb{R}^+$ is continuous. Hence
\begin{equation}
\begin{aligned}
\varepsilon^{\mathrm{cl}, R_\lambda}([\vec{h}_\lambda]_{P^{\oplus\ell}}) - \qty[f]_S &= \mathrm{eff}([\vec{h}_\lambda]_{P^{\oplus\ell}})^{-2 \ell} \; \tilde{\mathcal{E}}_\ell(\lambda),
\end{aligned}
\end{equation}
for $\lambda \mapsto \tilde{\mathcal{E}}_\ell(\lambda):=\mathrm{eff}([(h, \lambda h, ..., \lambda h)^T]_{P^{\oplus\ell}})^{2 \ell} \;  \qty[\mathcal{E}_\ell(\lambda)]_{S} \in \mathcal{C}_\mathcal{S}$, which is obviously continuous around $\lambda=0$. This concludes the proof of Theorem~\ref{theo_faster_approx}.
\end{proof}

We finally turn to the proof of Theorem~\ref{theo_faster_approx_fields}.

\begin{proof}[Proof of Theorem~\ref{theo_faster_approx_fields}]

Continuing in the notation from above let us define $\vec{g}_\lambda$ as the projection of $\theta_{\ell}(\lambda)\mqty(0 \\ \vec{h}_\lambda)$ to the last $\ell$ components and let us set $B_\lambda := \varphi_\mathcal{P}^{ \otimes \ell}(\vec{h}_\lambda) - \sigma(\varphi_\mathcal{P}^{ \otimes \ell}(\vec{g}_\lambda))\openone$ for some probe preparation state $\sigma$. Then it is easy to see that the collection of measurement schemes
\begin{equation}
    (\mathcal{F}_\mathcal{P}^{\otimes \ell}, \varepsilon_{\sigma}^{\varphi,\lambda R},{ B_\lambda})
\end{equation}
is an asymptotic measurement scheme for $\varphi_\mathcal{S}(f)$. Note that
\begin{equation}
    \varepsilon_{\sigma}^{\varphi,\lambda R} \qty({ B_\lambda}) - \varphi_\mathcal{S}(f) = \lambda^{2 \ell} \varphi_\mathcal{S}(\mathcal{E}_\ell(\lambda)).
\end{equation}
{ Let us now choose a seminorm $\mathrm{eff}$ such that $\mathrm{eff}\qty({ B_\lambda}) \neq 0$. (In particular, in the case when $\sigma$ is a product of quasi-free states with vanishing one-point function, then $\mathrm{eff}_\sigma$ as defined in Eq.~\eqref{eq_eff_sigma} is such a seminorm.)} We then have that
\begin{equation}
    \mathrm{eff}\qty({ B_\lambda}) = \lambda^{-1} \mathrm{eff}\qty(\varphi_\mathcal{P}^{ \otimes \ell}(\lambda \vec{h}_\lambda) - \sigma(\varphi_\mathcal{P}^{ \otimes \ell}(\lambda \vec{g}_\lambda))\openone),
\end{equation}
where $\lambda \mapsto \mathrm{eff}\qty(\varphi_\mathcal{P}^{ \otimes \ell}(\lambda \vec{h}_\lambda) - \sigma(\varphi_\mathcal{P}^{ \otimes \ell}(\lambda \vec{g}_\lambda))\openone) \in \mathbb{R}^+$ is continuous, since $\vec{g}_\lambda$ equals $\vec{h}_\lambda$ plus a continuous term, see Eq.~\eqref{eq_def_theta_j}. In summary 
\begin{equation}
    \varepsilon_{\sigma}^{\varphi,\lambda R} \qty({ B_\lambda}) - \varphi_\mathcal{S}(f) =  \mathrm{eff}\qty({ B_\lambda})^{-2 \ell} \tilde{\mathcal{E}}_\ell(\lambda),
\end{equation}
where $\lambda \mapsto \tilde{\mathcal{E}}_\ell(\lambda) := \mathrm{eff}\qty(\varphi_\mathcal{P}^{ \otimes \ell}(\lambda \vec{h}_\lambda) - \sigma(\varphi_\mathcal{P}^{ \otimes \ell}(\lambda \vec{g}_\lambda))\openone)^{2 \ell} \varphi_\mathcal{S}(\mathcal{E}_\ell(\lambda)) \in \mathcal{F}_\mathcal{S}$ is continuous on a neighbourhood of $\lambda=0$. Thus we have shown the theorem.
\end{proof}

\section{Abstract combination of asymptotic measurement schemes}\label{sec_appendix_abstract comb}

Let us consider the task of combining asymptotic measurement schemes at an abstract level. For $j=1,2$, let 
\begin{equation}
    H^j_\alpha:=\qty(\mathcal{P}^j_\alpha,\varepsilon_{\alpha,\sigma^j_\alpha}^j, B^j_\alpha)
\end{equation}
be asymptotic measurement schemes for $A_j \in \mathcal{S}$. Then
\begin{equation}
    H_\alpha:=\qty(\mathcal{P}^1_\alpha \otimes \mathcal{P}^2_\alpha,\varepsilon_{\alpha, \sigma^1_\alpha \otimes \sigma_\alpha^2}, c_1 B^1_\alpha \otimes \openone + c_2 \openone \otimes B^2_\alpha),
\end{equation}
is an asymptotic measurement scheme for $c_1 A_1 + c_2 A_2$,
if 
\begin{equation}
\begin{aligned}
    \varepsilon_{\alpha, \sigma^1_\alpha \otimes \sigma_\alpha^2}(B^1_\alpha \otimes \openone) &\to A_1,\\
    \varepsilon_{\alpha, \sigma^1_\alpha \otimes \sigma_\alpha^2}(\openone \otimes B^2_\alpha) &\to A_2,
\end{aligned}
\end{equation}
and if the topology on $\mathcal{S}$ respects addition. 

This situation occurs under some quite natural conditions provided that, for each $\alpha$, the 
coupling zones for the two probes $\mathcal{P}^1_\alpha$ and $\mathcal{P}^2_\alpha$ are \emph{causally orderable}, i.e., may be separated by a Cauchy surface, and that the probes may be
regarded as a single `super-probe' $\mathcal{P}^1_\alpha\otimes\mathcal{P}^2_\alpha$ that respects \emph{bipartite causal factorization}. The latter condition requires that the scattering map for the system coupled to the `super-probe' factorizes into the composition of the scattering maps for the two individual probes with the `later' probe first, e.g., $ \Theta_\alpha = \hat{\Theta}_\alpha^1 \circ \hat{\Theta}_\alpha^2$ if the second coupling region lies to 
the future of a Cauchy surface separating it from the first. (Here the
hat on $\Theta_\alpha^j$ indicates the trivial extension from
$\mathcal{S}\otimes\mathcal{P}^j_\alpha$ to $\mathcal{S}\otimes\mathcal{P}^1_\alpha\otimes\mathcal{P}^2_\alpha$.)
Crucially, the factorization must hold for all admissible causal ordering of the two coupling regions. See~\cite{bostelmann2020impossible} for a full discussion and the extension to multiple probes.

In particular, if the coupling zones for $\mathcal{P}^1_\alpha$ and $\mathcal{P}^2_\alpha$ are spacelike separated, then their individual scattering maps must commute on $\mathcal{S} \otimes \mathcal{P}^1_\alpha\otimes\mathcal{P}^2_\alpha$, because they admit both possible causal orderings. It follows that the resulting induced observable maps $\varepsilon_{\alpha, \sigma^1_\alpha \otimes \sigma_\alpha^2}$ fulfill
\begin{equation}
    \varepsilon_{\alpha, \sigma^1_\alpha \otimes \sigma_\alpha^2}(C^1_\alpha \otimes C^2_\alpha) = \varepsilon_{\alpha,\sigma^1_\alpha}^1(C^1_\alpha) \varepsilon_{\alpha,\sigma^2_\alpha}^2(C^2_\alpha),
\end{equation}
for every $C^j_\alpha \in \mathcal{P}^j_\alpha$, see Eqs.~(41)-(42) in~\cite{bostelmann2020impossible}. In particular,
\begin{equation}
     \varepsilon_{\alpha, \sigma^1_\alpha \otimes \sigma_\alpha^2}(B^1_\alpha \otimes \openone) = \varepsilon_{\alpha,\sigma^1_\alpha}^1(B^1_\alpha)
     \to A_1, \qquad \varepsilon_{\alpha, \sigma^1_\alpha \otimes \sigma_\alpha^2}(\openone \otimes B^2_\alpha) = \varepsilon_{\alpha,\sigma^2_\alpha}^2(B^2_\alpha)
     \to A_2,
\end{equation}
as required. This is an abstract solution to the problem of combining two asymptotic measurement schemes with spacelike separated coupling zones for observables $A_j$ to one asymptotic measurement scheme for an observable $c_1 A_1 + c_2 A_2$.

More generally, suppose that the coupling zones may be ordered so that the second coupling region lies to the future of a Cauchy surface separating it from the first (but not necessarily vice versa). Then the combined scattering maps factorize as
\begin{equation}
    \Theta_\alpha = \hat{\Theta}_\alpha^1 \circ \hat{\Theta}_\alpha^2,
\end{equation}
and the resulting map $\varepsilon_{\alpha, \sigma^1_\alpha \otimes \sigma_\alpha^2}$ fulfills for every $C^j_\alpha \in \mathcal{P}^j_\alpha$
\begin{equation}
     \begin{aligned}
         \varepsilon_{\alpha, \sigma^1_\alpha \otimes \sigma_\alpha^2}(C^1_\alpha \otimes \openone) &=\varepsilon_{\alpha,\sigma^1_\alpha}^1(C^1_\alpha),\\
         \varepsilon_{\alpha, \sigma^1_\alpha \otimes \sigma_\alpha^2}(\openone \otimes C^2_\alpha) &=\eta_{\sigma_\alpha^1} \qty(\Theta^1_\alpha \varepsilon_{\alpha,\sigma^2_\alpha}^2(C^2_\alpha) \otimes \openone).
     \end{aligned}
\end{equation}
Let us now assume that the maps $A \mapsto \eta_{\sigma_\alpha^1} \qty(\Theta^1_\alpha A \otimes \openone)$ on $\mathcal S$ converge pointwise to the identity\footnote{This is motivated by the specific asymptotic measurement schemes for $\mathcal{A}_\mathcal{S}$ above, in which the coupling becomes weaker along the net.} $\mathrm{id}$ and are equicontinuous\footnote{This holds in particular if the requirements for the uniform boundedness principle are met, see for instance Theorem~33.1 in~\cite{Treves:TVS}.} at $0$ so that $\eta_{\sigma_\alpha^1} \qty(\Theta^1_\alpha A \otimes \openone)\to A$ for all $A$ and $\eta_{\sigma_\alpha^1} \qty(\Theta^1_\alpha A_\alpha \otimes \openone)\to \lim_\alpha A_\alpha$ for all convergent $(A_\alpha)_\alpha$.
Then
\begin{equation}
     \begin{aligned}
         \lim\limits_\alpha \varepsilon_{\alpha, \sigma^1_\alpha \otimes \sigma_\alpha^2}(\openone \otimes B^2_\alpha) &=  \lim\limits_\alpha \eta_{\sigma_\alpha^1} \qty(\Theta^1_\alpha \varepsilon_{\alpha,\sigma^2_\alpha}^2(B^2_\alpha) \otimes \openone) \\
         & =  \lim\limits_\alpha \qty( \eta_{\sigma_\alpha^1} \qty(\Theta^1_\alpha A_2 \otimes \openone) + \eta_{\sigma_\alpha^1} \qty(\Theta^1_\alpha ( \varepsilon_{\alpha,\sigma^2_\alpha}^2(B^2_\alpha) -A_2) \otimes \openone))\\
         &= A_2 + \lim\limits_\alpha \eta_{\sigma_\alpha^1} \qty(\Theta^1_\alpha ( \varepsilon_{\alpha,\sigma^2_\alpha}^2(B^2_\alpha) -A_2) \otimes \openone)\\
         & =A_2,
     \end{aligned}
\end{equation}
where the last step follows by equicontinuity of the maps $A \mapsto \eta_{\sigma_\alpha^1} \qty(\Theta^1_\alpha A \otimes \openone)$
and because $\varepsilon_{\alpha,\sigma^2_\alpha}^2(B^2_\alpha) \to A_2$. Hence, under the stated assumptions, $(H_\alpha)_\alpha$ is an asymptotic measurement scheme for $c_1 A_1 + c_2 A_2$.\\

When the coupling zones are \emph{not} causally orderable we face the general problem of combining two theories, each describing a probe coupled to the system. While it is unclear how to do this for general theories,  it may be accomplished easily in a Lagrangian formulation by taking
the sum of the two coupled Lagrangians minus the system Lagrangian. This is precisely what we have done in \Sect\ref{sec_comb_measurement_schemes}.

\section{States, GNS representation and field operators}
\label{appendix_conv_field_ops}

Recall that a state $\omega:\mathcal{A} \to \mathbb{C}$ is a linear map such that $\omega(\openone)=1$, and $\forall A \in \mathcal{A}: \omega(A^* A) \geq 0$.

A state on the $C^*$-algebra of a scalar field $\mathcal{A}$ is called \emph{analytic}, if for every $f \in C_c^\infty(M;\mathbb{R}^k)$ the function $\mathbb{R}\ni t \mapsto \omega(W(tf)) \in \mathbb{C}$ is analytic. 

A useful class of physically reasonable analytic states on $\mathcal{A}$ is that of \emph{quasi-free} states. A state $\omega: \mathcal{A} \to \mathbb{C}$ is called quasi-free (or Gaussian, with vanishing one-point function), if
\begin{equation}
    \omega(W(f)) = e^{- \frac{1}{4} \beta(f,f)},
\end{equation}
for $\beta: C_c^\infty(M; \mathbb{R}^k) \times C_c^\infty(M; \mathbb{R}^k) \to \mathbb{R}$ a symmetric $\mathbb{R}$-bilinear form such that there exists a positive semidefinite $\mathbb{R}$-bilinear form $\tilde{\beta}: C_c^\infty(M; \mathbb{R}^k)/ P C_c^\infty(M; \mathbb{R}^k) \times C_c^\infty(M; \mathbb{R}^k)/P C_c^\infty(M; \mathbb{R}^k)$ with $\beta(f,g) = \tilde{\beta}([f]_P,[g]_P)$. The positivity of $\omega$ is equivalent to
\begin{equation}
    \begin{aligned}
    |E(f,g)|^2 \leq \beta(f,f) \beta(g,g).
    \label{eq_positivity_condition}
    \end{aligned}
\end{equation}
We see in particular that $\omega(W(f)) \neq 0$. As examples, we mention that the QFT of a linear scalar field on a stationary globally hyperbolic spacetime, subject to a positive stationary potential, admits quasi-free ground and KMS states with distributional $\beta$, see~\cite{sanders2013thermal}.

For a given (not-necessarily quasi-free) state $\omega$, the famous GNS construction allows us to represent the algebra $\mathcal{A}$ as bounded operators on a complex (not necessarily separable) Hilbert space $\qty(\mathcal{H}_\omega,\braket{\cdot |\cdot})$, i.e., there is a $C^*$-homomorphism $\pi_\omega: \mathcal{A} \to BL(\mathcal{H}_\omega)$ called the \emph{GNS representation of $\omega$}. Furthermore, there exists $\Omega_\omega \in \mathcal{H}_\omega$ such that $\omega(A) = \braket{\Omega_\omega|\pi_\omega(A) \Omega_\omega}$ and which is cyclic for $\pi_\omega \qty[\mathcal{A}]$, i.e., $\{A\Omega_\omega | A\in \pi_\omega\qty[\mathcal{A}]\}$ is dense in $\mathcal{H}_\omega$. Since $\pi_\omega$ is continuous, it even holds that for every dense sub-algebra $\tilde{\mathcal{A}}$ of $\mathcal{A}$, $\pi_\omega\qty[\tilde{\mathcal{A}}]\Omega_\omega$ is dense in $\mathcal{H}_\omega$.

It is worth emphasising that $\omega$ is represented by the vector state $\Omega_\omega \in \mathcal{H}_\omega$ in the GNS representation whether or not $\omega$ is a pure state (purity holds if and only if the representation $\pi$ is irreducible).\\

Next,{ recall that $\mathfrak{S}_c$ is the set of all states $\omega$, such that $\omega \circ W$ is continuous. While the map $f \mapsto W(f)$ is not continuous in the norm topology of $\mathcal{A}$, as the norm-distance between two Weyl generators indexed by any two inequivalent functions is two, we will now show that the map $\pi_\omega \circ W$ is continuous with respect to the strong$^*$ operator topology on $BL(\mathcal{H}_\omega)$, for $\omega\in\mathfrak{S}_c$. 

\begin{lem}
For every $\omega \in \mathfrak{S}_c$ with GNS representation $\pi_\omega: \mathcal{A} \to BL(\mathcal{H}_\omega)$, the map
\begin{equation}
    \pi_\omega \circ W : C_c^\infty(M;\mathbb{R}^k) \to BL(\mathcal{H}_\omega) 
\end{equation}
is continuous with respect to the strong$^*$ operator topology on $BL(\mathcal{H}_\omega)$.
\label{lem_strong_convergence_id}
\end{lem}

\paragraph*{Remark:} The assumptions of Lemma~\ref{lem_strong_convergence_id} are fulfilled in particular for quasi-free states with distributional $\beta$.}

\begin{proof}
Let us set $W^\pi(f):= \pi_\omega(W_\mathcal{S}(f))$ and let $\qty(f_n)_n$ be a net in $C_c^\infty(M;\mathbb{R}^k)$ that converges to $f$. Using the Weyl relations it is easy to verify that
\begin{equation}
    \qty(W(f_n) - W(f))W(g) = e^{- \frac{\mathrm{i}}{2} E(f_n-f, f+2g)} W(f) W(g) \qty( W(f_n - f)- e^{ \frac{\mathrm{i}}{2} E(f_n-f, f+2g)} \openone),
\end{equation}
and hence, by unitarity of the Weyl generators,
\begin{equation}
\begin{aligned}
  \|(W^\pi(f_n)-W^\pi(f))W^\pi(g) \Omega_\omega\|_\omega &= \|W(f_n - f)\Omega_\omega- e^{ \frac{\mathrm{i}}{2} E(f_n-f, f+2g)}\Omega_\omega \|_\omega \\
  & \leq \| (W(f_n - f) - \openone) \Omega_\omega\|_\omega + |e^{ \frac{\mathrm{i}}{2} E(f_n-f, f+2g)} -1|\\
  &\longrightarrow 0,
\end{aligned}
\end{equation}
as $f_n\to f$, on noting that $\| (W^\pi(f_n-f)-\openone)\Omega_\omega\|_\omega^2= 2-2\text{Re}\, \omega(W(f_n-f))$, and using continuity of $\omega\circ W$ and $\omega(W(0))=1$, together with the distributional nature of $E$. Taking linear combinations, we have shown that $W^\pi(f_n)\phi\to W^\pi(f)\phi$ for all $\phi$ in the span of $\{W^\pi(g)\Omega: g\in C_c^\infty(M)\}$, which is dense in $\mathcal{H}_\omega$, due to cyclicity and the Weyl relations. This statement extends to all $\phi\in \mathcal{H}_\omega$ because the $W^\pi(f_n)$ are unitary and therefore uniformly bounded.

Finally, if $f_n \to f$, then also $-f_n \to -f$ and the above argument shows that for every $\phi \in \mathcal{H}_\omega$ it holds that
\begin{equation}
\begin{aligned}
    &\|\qty(W^\pi(f_n) - W^\pi(f)) \phi\|_\omega  + \|\qty(W^\pi(f_n) - W^\pi(f))^* \phi\|_\omega\\
    &=  \|\qty(W^\pi(f_n) - W^\pi(f)) \phi\|_\omega  + \|\qty(W^\pi(-f_n) - W^\pi(-f)) \phi\|_\omega \longrightarrow 0.
\end{aligned}
\end{equation}
\end{proof}

Let us now turn to a discussion of field operators, starting with some prerequisites based on Chapter 3 of~\cite{petz1990invitation}. Let $\pi$ be the GNS representation of an analytic state $\omega$ of the CCR-$C^*$-algebra $\mathcal{A}$, with GNS vector $\Omega$ and Hilbert space $\mathcal{H}$. Then for every $f \in C_c^\infty(M;\mathbb{R})$ the map $t \mapsto \pi(W(tf))$ is a strongly continuous one-parameter group which allows us to define field operators $\varphi^\pi(f)$ via Stone's theorem. For every $f \in C_c^\infty(M;\mathbb{R})$, $\varphi^\pi(f)$ is a self-adjoint operator with dense domain $D(\varphi^\pi(f))$ such that $\pi(W(f)) =e^{\mathrm{i}\varphi^\pi(f)}$. 
The functions $(f_1, ..., f_n) \mapsto \braket{\Omega| \varphi^\pi(f_1)...\varphi^\pi(f_n)\Omega}$ ($n\in\mathbb{N}$) are called $n$-point functions; in particular, the \emph{two-point function} takes the form $\braket{\Omega|\varphi^\pi(f) \varphi^\pi(g) \Omega}=\frac{1}{2} \beta(f,g) + \frac{\mathrm{i}}{2}E(f,g)$ 
for quasi-free $\omega$ with bilinear form $\beta$.

It follows indeed from the proof of Lemma~\ref{lem_strong_convergence_id} that the subspace $D^\omega_W := \mathrm{span} \{ \pi(W(f)) \Omega | f \in C_c^\infty(M;\mathbb{R})\}$ is contained in $D(\varphi^\pi(f))$ for every $f$, and is dense in the GNS Hilbert space $\mathcal{H}$. Since $D^\omega_W$ is trivially invariant under the one-parameter unitary group $t \mapsto e^{\mathrm{i} t \varphi^\pi(f)}$ for every $f$, we get from Nelson's invariant domain theorem, see Theorem VIII.10 in~\cite{ReedSimonI1980}, that $\varphi^\pi(f) \restriction D^\omega_W$ is essentially self-adjoint. In particular, $D^\omega_W$ is a common core for all $\varphi^\pi(f)$.\\

Furthermore, for every \emph{quasi-free} state $\omega \in \mathfrak{S}_c^{an}$ the space $D^\omega_\varphi:= \mathrm{span} \{ \varphi^\pi(g_1) ... \varphi^\pi(g_n) \Omega| n \in \mathbb{N}; g_1, ..., g_n \in C_c^\infty(M;\mathbb{R)}\}$ is dense in the GNS Hilbert space $\mathcal{H}$ and consists of \emph{analytic vectors} for every $\varphi^\pi(f)$, so is in particular a common core, see Corollary 4.10 in~\cite{petz1990invitation}.

In fact, the complex $*$-algebra spanned by operators $\varphi^\pi(f)$ on $D_\varphi^\omega$ for $f \in C_c^\infty(M;\mathbb{R})$ is a representation of the field-$*$-algebra of \Sect\ref{sec:field_algebra}. The property that (for quasi-free $\omega$) $D_\varphi^\omega$ is dense in $\mathcal{H}$ is equivalent to the fact that $\Omega$ is a cyclic vector for the field-$*$-algebra. In other words, the GNS representation of a quasi-free state on the $C^*$-algebra $\mathcal{A}$ carries a representation of the field-$*$-algebra that is equivalent to the GNS representation of the field-$*$-algebra in the quasi-free state with the same two-point function.

\end{document}